\documentclass[journal]{IEEEtran}

\ifCLASSINFOpdf
\else
   \usepackage[dvips]{graphicx}
\fi
\usepackage{url}
\usepackage{amsmath}
\usepackage{graphicx}
\usepackage{placeins}
\usepackage{float}
\usepackage{bbm}
\usepackage{subfigure}
\usepackage{verbatim}
\usepackage{amssymb}
\usepackage{amsmath}
\usepackage{amsthm}
\usepackage{mwe}
\usepackage{docmute}
\usepackage{mathtools}
\usepackage{siunitx}
\usepackage[ruled, vlined]{algorithm2e}
\usepackage[switch]{lineno}

\usepackage[disable]{todonotes}

\usepackage{graphicx}

\DeclareMathOperator*{\argmin}{argmin}

\newtheorem{definition}{Definition}
\newtheorem{lemma}{Lemma}

\newtheorem*{theorem*}{Theorem}

\theoremstyle{definition}

\newtheorem{assumption}{Assumption}

\newcommand{\revision}[1]{\textcolor{black}{#1}} 
\newcommand{\rev}[1]{\textcolor{black}{#1}}

\begin{document}

\title{Timely Target Tracking: Distributed Updating in Cognitive Radar Networks}
\author{William W. Howard, Anthony F. Martone, R. Michael Buehrer
\thanks{W.W. Howard and R.M. Buehrer are with Wireless@VT, Bradley Dept. of Electrical and Computer Engineering, Virginia Tech, Blacksburg, VA 24060, USA. \\ 
Correspondence:$\{${wwhoward}$\}$@vt.edu  \\
A.F. Martone is with DEVCOM Army Research Laboratory, Adelphi, MD 20783. (e-mail:{anthony.f.martone.civ}@army.mil).\\
Distribution Statement A: Approved for public release. Distribution is unlimited. \\
Portions of this work were presented at IEEE Radar Conference 2023, San Antonio, TX \cite{howard2023_timelyconf}. }
}

\maketitle
\pagenumbering{gobble}

\begin{abstract}
Cognitive radar networks (\textbf{CRNs}) are capable of optimizing operating parameters in order to provide actionable information to an operator or secondary system. 
CRNs have been proposed to answer the need for low-cost devices tracking potentially large numbers of targets in geographically diverse regions. 
Networks of small-scale devices have also been shown to outperform legacy, large scale, high price, single-device installations. 
In this work, we consider a CRN tracking multiple targets with a goal of providing information which is both fresh and accurate to a measurement fusion center (\textbf{FC}).  
We show that under a constraint on the update rate of each radar node, the network is able to utilize Age of Information (\textbf{AoI}) metrics to maximize the resource utilization and minimize error per track. 
Since information freshness is critical to decision-making, this structure enables a CRN to provide the highest-quality information possible to a downstream system or operator. 
We discuss centralized and distributed approaches to solving this problem, taking into account the quality of node observations, the maneuverability of each target, and a limit on the rate at which any node may provide updates to the FC. 
We present a centralized AoI-inspired node selection metric, where a FC requests updates from specific nodes. 
We compare this against several alternative techniques. 
Further, we provide a distributed approach which utilizes the Age of Incorrect Information (\textbf{AoII}) metric, allowing each independent node to provide updates according to the targets it can observe. 
We provide mathematical analysis of the rate limits defined for the centralized and distributed approaches, showing that they are equivalent. 
We conclude with numerical simulations demonstrating that the performance of the algorithms exceeds that of alternative approaches, both in resource utilization and in tracking performance. 
\end{abstract}


\section{Introduction}
\subsection{Cognitive Radar Networks}
Cognitive radar networks (\textbf{CRNs}) consist of many cognitive radar devices (``nodes'') which observe targets and report to a fusion center (\textbf{FC}). 
The literature has provided many cognition models for CRNs including completely distributed cognition \cite{howard2021_multiplayerconf} \cite{howard2022_MMABjournal} \cite{howard2022_adversarialconf} and centralized cognition \cite{howard2023_hybridjournal} \cite{howard2022_decentralized_conf}. 
A topic of study for nearly two decades, the primary benefits of CRNs over traditional networked radar include spatial diversity, flexibility in dynamic environments, efficient usage of spectrum, and increased resilience to outages \cite{haykin2006}.

CRNs exist to gather information from targets in a possibly geographically diverse region. 
Cognitive radio was first proposed to address the need for communication systems to be adaptive in congested and possibly contested spectrum. 
The term ``cognitive'' appeals to the biological decision-making cycle as described variously \cite{Martone_CRN_loop} \cite{6218166} \cite{7910111}. 
In particular, cognitive radios employ the ``perception-action cycle,'' observing their environment and taking actions which are predicted by some internal model to positively impact performance. 
The current work employs cognitive processing to evaluate dynamic target behavior, reduce utilization of a shared communication resource, and improve the quality of FC estimation when using a network of radars.

Other works have conducted research in the area of radar network \emph{resource allocation}. 
Several works consider the application of the Bayesian Cram\'er-Rao Lower Bound \textbf{(BCRLB)} \cite{van1968detection}, which is a bound on the variance of any unbiased estimator. 
The BCRLB is a very useful tool in the radar context is at allows a device to estimate the expected error for proposed parameters. 
A specific implementation of the BCRLB as an optimization criteria has been in the field of \emph{power allocation}. 
Yan et al. consider the use of the BCRLB and develop a power allocation scheme for a radar sensor network \cite{YAN2020173} \cite{9133517}. 
They consider constrained optimization of network power subject to a minimum tracking error, determined by the BCRLB. 
They also provide a survey of resource allocation techniques for radar sensor networks \cite{YAN2022104}. 
\rev{The BCRLB differs from the metrics we propose in this work. 
Rather than predicting the future performance, our AoI-based metrics use the time since observations were received at the FC. 
This allows a low computational cost method of determining when nodes should provide updates. }

Rather than seeking to allocate resources between heterogeneous radars in a network, which each may have different operating parameters and performance, we assume that our network is composed of similar cognitive radar nodes. 
\rev{Rather than centrally selecting a set of radar devices based on expected performance via the BCRLB, we allow radar nodes to determine the time steps at which updates should be provided.} This means that problems such as power allocation will not have the same trade-offs as in a heterogeneous network. 
We study the allocation of communication resources in order to optimize the timeliness of updates. 
Further, the prior work does not consider the impact of heterogeneous \emph{targets} (in terms of motion model) on the resource optimization nor communication overhead required to achieve these gains. 
We seek to fill that void through this work.

\emph{Sensor collaboration} \cite{985685} is the process of collecting target information from low-power ad-hoc distributed sensor networks, often using graph theory to determine the optimal path for sensor information or for sensor querying. 
This style of target tracking differs from our problem. 
Instead of selecting the sensors (radar nodes) which should make target measurements, we are optimizing the time steps at which radar nodes send updates to a central fusion center.

\rev{Markov decision processes (\textbf{MDPs}) and their generalization to partially observed Markov decision processes (\textbf{POMDPs}) are mathematical structures often used to analyze decision-theoretic problems. 
Their relevance to our work is due to the system which we consider, where a network sequentially takes some action and observes some change in state. 
MDPs and POMDPs have been used to address resources allocation in radar networks previous to this work. 
The survey of \cite{4205087} identifies four key areas of application for MDP models in distributed multi-target tracking problems: power, sensing, communication, and computation. 
Our chosen problem falls into the third category of communication constraints. 
As we will discuss, our proposed solution adapts an age-of-information metric from the literature, which is shown to be a Bellman-optimal solution to an MDP in \cite{AgeOfIncorrectInformation}. }
\subsection{Single Node Techniques}

Single-node cognitive radar (i.e., not networked) has been investigated even more in the literature. 
Examples of proposed problems include waveform design optimization \cite{5074349} \cite{thornton2020efficient} \cite{kirk2017cognitive} \cite{thornton2022_universaljournal}, dynamic spectrum access (\textbf{DSA}) \cite{kirk2018avoidance} \cite{ravenscroft2018experimental} \cite{kovarskiy2020spectral}, and improved target detection \cite{martone2017adaptable} \cite{thornton2023icassp}. 
Commonly, cognitive radar problems include a reliance on \emph{online optimization} - sequential processes through which performance is improved. 
Online learning is more effective because offline learning is ineffective without high quality prior knowledge.

\revision{
Another branch of the single-node cognitive radar investigates the performance of co-located multiple-input multiple-output (\textbf{MIMO}) radar systems using multiple beams for target observation \cite{6960029} \cite{7069270}. 
The authors use convex optimization to separate the nonconvex allocation problem into several convex problems which may be solved optimally, and show that the worst case single-node tracking accuracy can be significantly improved. 
}

In particular, there has been a study of meta-cognitive techniques in single-node cognitive radar \cite{9114775}, \cite{10007921}, where a higher-level process analyzes the performance of different cognitive agents and adaptive selects those agents which exhibit superior performance. 
The work of \cite{thornton2022cognitive} analyzes the specific trade-off between a learning-based approach to DSA in cognitive radar and a fixed rule-based approach, concluding that there are regimes in which either technique is dominant. 

\revision{
\emph{Multi-target tracking} is the segment of the literature which addresses problems of data association and track maintenance \cite{5744132}. 
Data association \cite{4441756} becomes necessary when a sensor receives measurements of multiple targets (with some probability of missed detections and some probability of false alarms) and must \emph{associate} each measurement with a target track. 
The probability hypothesis density (\textbf{PHD}) filter \cite{1261119} \cite{5730505} \cite{5259179} has been proposed as a method of interpreting multi-target sensor data and forming target tracks. 
The PHD filter frames detections as random finite sets and uses the first-order moment to determine an intensity function over the target state space. 
A tractable PHD filter implementation was demonstrated in \cite{1710358}, which uses Gaussian mixture (\textbf{GM-PHD}) models in the first closed-form solution to PHD filtering. 
Since in our work target association is performed at the individual radar nodes and is not dependent on update frequency or update times, target association performance will have limited impact on the performance of node selection policies. 
Similarly, since we model target birth and death as isotropic processes (i.e., having constant probability throughout space and time), the probability of track initiation has no impact on the node selection algorithms. }

\revision{
Our work focuses on downstream systems from multi-target tracking. 
In particular, we consider questions such as ``How often should the FC receive updates on each target?'' and ``How do nodes decide when to send updates?''. 
}

\vspace{0.1in}
\subsection{Age of Information}
First proposed in \cite{AoI_initial} and gaining considerable traction recently, Age of Information (\textbf{AoI}) tools are popular when information freshness is desired. 
The survey by Yates et al. \cite{AoI_survey2} covers recent contributions and applications, characterizing AoI as ``performance metrics that describe the timeliness of a monitor's knowledge of an entity or process.'' 

Information freshness can be quantified via an ``age process.''  
Typically, an age process $\Delta(t)$ increases linearly in time (i.e., aging at a rate of one second per second). 
In the literature from which we draw our approaches, age processes relate to the time since an aggregator received updates from an observer on a process. 
So, when the observer provides an update to the aggregator, the age of the process is reset to 0. 

AoI has been implemented particularly often to solve problems in the field of \emph{federated learning}, where a central parameter server attempts to train a large machine learning (\textbf{ML}) model using numerous independent clients. 
Federated learning is most important in those domains where \emph{data privacy}\footnote{The distinction between data privacy and security is that while secure systems may share data, data privacy systems may not share data. } is essential, e.g. medicine or other personal identifiable information. 
Information freshness is utilized in this field to ensure the global model is updated by the most recent data. 
Our current work does not have the same purpose. 
Information freshness remains critical in CRNs to provide the most recent and accurate data to operators or downstream systems, and data security remains important, but not data privacy. 

Another field frequently finding use for AoI is distributed sensor networking, where multiple devices observe one or many processes. 
In \cite{AoI_Dhillon}, the authors describe a network of UAVs which assist an IoT-enabled network utilizing an AoI metric to maximize information freshness. 
In this and several similar works \cite{AoI_sensor}, \cite{AoI_multiple}, a scheduler must assign resources to each of several nodes. 
The first part of our work is similar to these, where we develop a centralized decision metric. 
However, the second part of our work derives a distributed technique which differs from these. 

The Age of Incorrect Information (\textbf{AoII}) \cite{AgeOfIncorrectInformation} is a recent metric which proposes to ``extend the notion of fresh updates to that of fresh \emph{informative} updates.''
Rather than measuring the time from the last update, AoII considers the information content of updates and in particular those which bring \emph{new and correct} information to the aggregator. 
This recognizes the fact that Markovian sources may not have new states at all times, so an observer might not need to continuously update an aggregator. 
A new update which contains identical information to the previous update is not necessarily useful. 
AoII has been utilized several times for centralized tracking of remote sources \cite{AoII_ISIT}, \cite{AoII_infocom}, but to the best of our knowledge it has not been adapted to the distributed, multi-process tracking problem. 

\subsection{Problem Summary}
Due to the possibly large and diverse geographic region under observation by the CRN, it follows that there may be nodes which see no targets and other nodes which see many targets. 
Similarly, there may be targets which are seen by many nodes and targets which are seen by no nodes. 
The quality of each node's observations of a given target (\emph{target observability}) will be influenced by large-scale channel conditions such as path length and small-scale channel conditions such as target frequency selectivity and angle. 
These effects are greatly impacted by the network geometry - the physical distances between radar nodes and their interactions with the environment and targets. 
We utilize a model driven by \emph{stochastic geometry} to describe these effects.

As the scene evolves in time, targets may exhibit dynamic motion, moving unpredictably. 
Similarly, the number of targets existing in the region need not stay fixed; targets are able to enter and exit the region. 
We model the targets as uncrewed aerial vehicles (\textbf{UAVs}) with independent and identically distributed (\textbf{i.i.d.}) motion models, which may take off or land from anywhere in the considered region. 
The motion models we consider are discrete-time Markov chains, where the probability of a transition to a different state of motion (e.g. linear or turning) depends only on the current state. 
We'll discuss a centralized polling process by which a node can communicate to the FC whether a target in their region has exhibited ``interesting'' behavior, such as a change in motion model state. 
Then, we will present a distributed process through which a node can decide whether or not to provide an update to the FC without any centralized intervention, utilizing information from target motion modeling. 
We refer to this stochastic target motion as \emph{maneuverability}.

The AoI inspired approaches we discuss in this and in our previous work \cite{howard2023_timelyconf} are intended to minimize the amount of time that the aggregator (FC) has out-of-date information on target processes (subject to constraints), and thereby minimize the error of the target state estimate which is provided to an operator or downstream system. 
If there are sufficient resources then the obvious and optimal solution is to send updates in every possible instance. 
However, as spectrum is inherently scarce, it is impractical and inefficient for the nodes to update the FC state in every possible instance. 
So, we impose an \emph{update rate constraint} on each node in order to keep the average utilized communication rate below a limit. 
The update rate is constrained due to congested spectrum; as we will describe, the network must share spectrum resources with other devices and is limited to some fraction of the available resources. 

We discuss how the consideration of target observability, maneuverability, and this update rate constraint impact the development of our AoI metric. 
The primary \emph{tool} we'll use is the AoII \cite{AgeOfIncorrectInformation}, which approaches this type of problem using a Markov chain model, providing a Bellman-optimal update policy. 
The primary \emph{modifications} we must make to AoII are consequences of scale; AoII was developed for single-process, single-observer systems and we consider the case where multiple nodes observe multiple targets. 
We measure the effectiveness of our proposed techniques using tools from the AoI literature as well as by calculating the tracking error achieved by the FC. 
Since this problem has not been addressed in the literature, we contextualize our results by comparison to likely candidates - namely, an approach based on iterative reinforcement learning (multi-armed bandits) and a random selection algorithm.

\subsection{Contributions} 
This problem, which has not yet been addressed in the literature save for our initial work \cite{howard2023_timelyconf}, resembles scheduling problems where a central server must collect information from distributed nodes. 
Consequently, we borrow from the AoI literature and develop two decision metrics. 
To the best of our knowledge, this work represents the first consideration of AoI metrics in CRNs. 
Specifically, we contribute the following: 
\begin{itemize}
    \item A centralized ``track-sensitive AoI metric,'' which utilizes a polling process to allow the FC to select nodes to provide updates in each update interval. 
    \item An adaptation of AoII to enable distributed decision-making, where each node implements a Bellman-optimal policy to determine when to provide updates. 
    In contrast to the centralized solution, where the FC coordinates all interactions, the distributed solution relies on each node to decide when to send updates. 
    Specific adaptations to the AoII algorithm include a modified Markov model and rate limits for multiple nodes. 
    \item Mathematical analysis of our proposed techniques. 
    \item Numerical simulations to support our conclusions. 
\end{itemize}

\subsection{Notation}
We use the following notation. 
Matrices and vectors are denoted as bold upper $\mathbf{X}$ or lower $\mathbf{x}$ case letters respectively.
Element-wise multiplication of two matrices or vectors is shown as $\mathbf{X}\odot \mathbf{Y}$. 
Functions are shown as plain letters $F$ or $f$. 
Sets $\mathcal{A}$ are shown as script letters. 
Denote the Lebesgue measure of a set $\mathcal{A}$ as $|\mathcal{A}|$. 
When we wish to show the number of elements in a (finite) set $\mathcal{A}$ rather than its measure, we use the cardinality $\#(\mathcal{A}$). 
The transpose operation is $\mathbf{X}^T$. 
The backslash $\mathcal{A}\backslash \mathcal{B}$ represents the set difference. 
Boxes (intervals) in $\mathbb{R}^d$ are written as $[a,b]^d$ and when the elements of a set are denoted, they are given as $\mathcal{A} = \{a, b, c, \dots \}$. 
Random variables are written as upper-case letters $X$, and their distributions will be specified as $X\sim\mathcal{D}(\cdot)$. 
The letter $U[a,b]$ is used to denote the uniform distribution on an interval $[a,b]$. 
The set of all real numbers is $\mathbb{R}$ and the set of integers is $\mathbb{Z}$. 
The speed of electromagnetic radiation in a vacuum is given as $c$. 
The Euclidean ($\ell^2$) norm of a vector $\mathbf{x}$ is written as $||\mathbf{x}||_2$. 
Estimates of a true parameter $p$ are given as $\hat{p}$. 

\subsection{Organization} 
The remainder of this work is organized as follows. In Section \ref{sec:models}, we discuss prior work in this area and develop relevant models. 
In Section \ref{sec:centralized} we derive a technique for centralized decision-making. 
Then, in Section \ref{sec:distributed} we discuss our AoII-based method for distributed updating. 
Section \ref{sec:simulations} provides numerical simulations and in Section \ref{sec:conclusions} we draw conclusions and suggest future work.


\section{System Model}
\label{sec:models}
In this section we will describe a general framework which will inform the development of our algorithms. 
We'll also cover mathematical preliminaries which will be required later. 
While these models make some specific assumptions (e.g., on the motion of targets), we strive to avoid any assumptions that would limit the applicability of our work. 

\subsection{Spatial Modeling}
\label{ss:spatial} 
Sensor networks (and more broadly wireless networks) are often analyzed using tools from \emph{stochastic geometry}. 
This is the field of study which considers the mathematical and statistical relationships between spatial processes. 
Poisson point processes (\textbf{PPPs}) are of particular interest, due to their generalizabililty and mathematical tractability. 
A point process is a random collection of points in space. 
As provided in \cite{haenggi}, a PPP on $\mathbb{R}^d$ with \emph{intensity measure} $\Lambda$ and \emph{density} $\lambda$ per unit area can be simulated on a compact region\footnote{Compactness is the only required condition, and the region can be disjoint. However, the procedure to simulate a PPP is most straightforward in rectangular regions} $B \subset \mathbb{R}^2$ with measure $|B|$ by drawing a Poisson number $N$ with mean $\lambda|B|$, and place $N$ points uniformly at random in $B$. 
Formally, 
\begin{equation}
    \label{eq:poisson}
    Pr(N=n) = \frac{\lambda|B|^n e^{-\lambda|B|}}{n!}
\end{equation}
In other words the number of points\footnote{\rev{By ``points'' we refer to ground-truth target locations rather than radar detections of targets. As covered by \cite{8289337}, target \emph{detections} and therefore the posterior hypothesis density will not be spatially uniform. }} inside a compact region is distributed according to the measure (or size) of that region. 

For numerous reasons (e.g., power limits, geographic obstructions and pulse range gating), the practical range of any radar node is limited. 
So, we model each radar node $n$ as covering a specific region $S_n$ with measure $|S_n|$. 
We will discuss later the fact that the FC does not need to know the shape of this region if $|S_n|$ is known.

\begin{figure}
    \centering
    \includegraphics[scale=0.55]{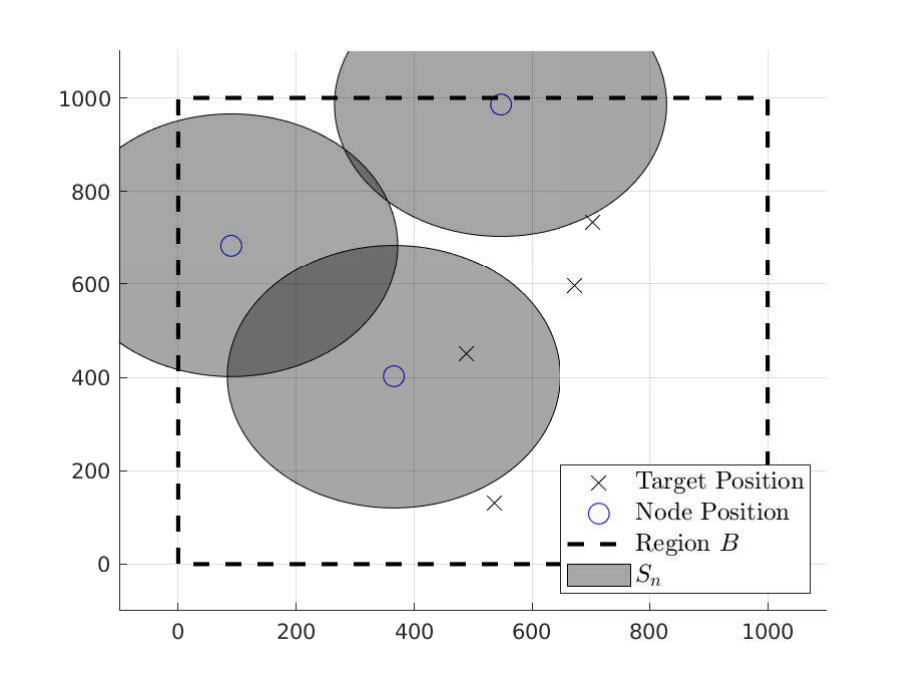}
    \caption{Tracking scenario with node density $\lambda_n=3$ and target density $\lambda_m=7$. }
    \label{fig:scene}
\end{figure}

Consider a compact region $B\subset \mathbb{R}^2$. 
Inside this region, targets are positioned according to a PPP with density $\lambda_m$. 
The target positions evolve in time according to the motion model discussed below. 
Further, inside the compact region $B$, we place a number of radar nodes according to a PPP with density $\lambda_n < \lambda_m$. 
Let $M$ be the Poisson random variable describing the number of targets, and note that $\overline{M}=\lambda_m|B|$ is the mean of this distribution. 
For a given instance of the model, we collect all of the targets into the set $\mathcal{M}$ and index them as $m\in\mathcal{M}$. 
When we wish to refer to the specific coordinates of a target $m$, it is denoted as $X_m$. 
Similarly, we collect each radar node $n$ into the set $\mathcal{N}$, letting $N$ with mean $\overline{N}$ be the random variable which describes the number of radar nodes. 
$X_n$ represents the coordinates of node $n$. 
When we conduct numerical simulations of stochastic processes, a single simulation consists of a single realization of each distribution. 
The resulting set of points is referred to as a point pattern. 

The \emph{coverage} of a wireless network is the \revision{union} of the coverage of each node. 
In general, each node $n\in\mathcal{N}$ covers a compact region $S_n$, which is a subset of $B$. 
The covered regions are drawn i.i.d. \revision{from} a spatial distribution\footnote{Commonly balls of random radii are used for i.i.d. coverage but in general $S_n$ need only be compact with $|S_n|>0$. }. 
In the typical \emph{Boolean model} \cite{haenggi}, it is particularly assumed that each node in a network covers a disk of fixed radius $r$. 
The probability that a location $x$ is covered by the network is then
\begin{align}
    \label{eq:coverage}
    Pr(x\in \bigcup_{n\in\mathcal{N}} S_n) &= 1-e^{-\lambda_n\mathbb{E}|S_n|}\\
    &= 1-e^{-\lambda_n\pi r^2}
\end{align}

This means that the probability any target is covered is given as Eq. (\ref{eq:coverage}). 
Further, for a single instance of this model, the number of radar nodes which cover a given target is Poisson distributed with mean $\lambda_n\pi r^2$.
Note that unlike the target spatial distribution, the coverage is constant in time since this is a general result for any location $x$. 
Since each target has a less than unity probability of being observed, some number of targets will be unobserved by the network in each time step. 
That probability can be found as Eq. (\ref{eq:unobserved}). 
\begin{equation}
    \label{eq:unobserved}
    \#\{m\notin \bigcup_{n\in\mathcal{N}}S_n\} \sim \text{Poisson}\left(\lambda_m e^{-\lambda_n \pi r^2}\right)
\end{equation}

Figure \ref{fig:scene} shows an instance of such a network, with a target density $\lambda_m = 20$ and a node density $\lambda_n = 10$. 
The coverage of each node is a circle centered on the node with an area of $|S_n| = 0.2|B|$.

\subsection{Motion Modeling}
\label{ss:motion}
We model the motion of the target UAVs according to a multi-state Markov chain \cite{rs12223789}. 
While in general many types of motion could be considered, in this work we focus on constant-velocity and constant-turn motion. 
When initialized, each target $m\in\mathcal{M}$ draws probabilities $P_{m,1}\sim U[0.7, 0.9]$ and $P_{m,2}\sim U[0.5, 0.7]$ to form a transition matrix: 
\begin{equation}
\label{eq:transition}
    T_m = \begin{bmatrix}
        P_{m,1} & 1-P_{m,1}\\
        1-P_{m,2} & P_{m,2}
        \end{bmatrix}
\end{equation}
Figure \ref{fig:markov_chain} demonstrates this Markov chain. 
The distributions over transition probability can be tailored such that one state is more common than others - specifically, such that constant velocity motion is much more common than constant turn motion. 

\begin{assumption}
\label{ass:motion}
The motion for each target evolves according to a Markov chain with constant transition probabilities, drawn i.i.d. for each node. The radar nodes must estimate these probabilities. 
\end{assumption}

\begin{figure}
    \centering
    \includegraphics[scale=0.15]{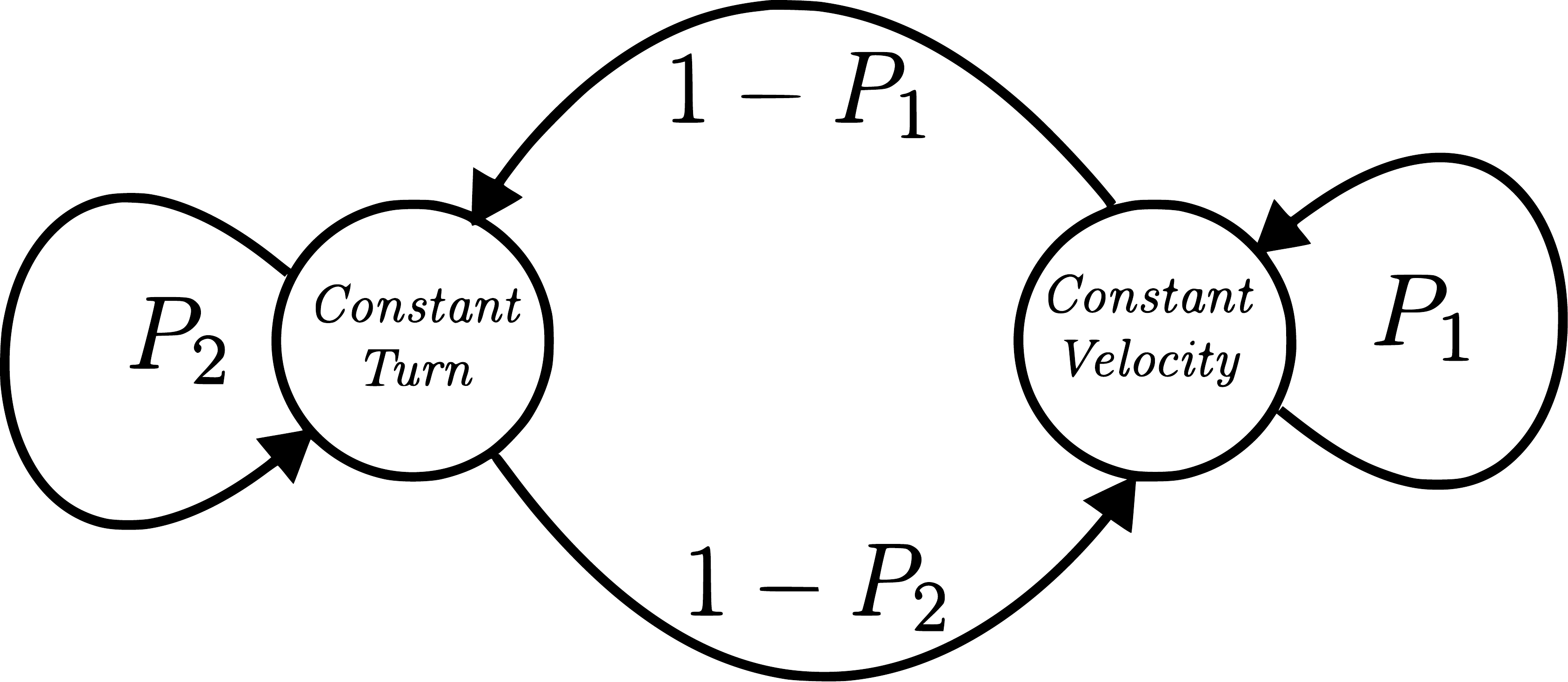} 
    \caption{The Markov chain motion model exhibited by UAV targets. }
    \label{fig:markov_chain}
\end{figure}

\subsection{Network Modeling}
As discussed, the network consists of the set of radar nodes $\mathcal{N}$ which attempt to track the set of targets $\mathcal{M}$, feeding the information to the FC. 
Figure \ref{fig:system} shows a version of this network: several targets are in the scene, with some observed multiple times and some not observed at all. 
The shortest time interval considered is the Coherent Pulse Interval (\textbf{CPI}), during which each radar node emits several pulses, coherently integrates them, and performs signal processing to extract estimated target parameters (position and velocity). 
\revision{T}he CPI duration for each node $n$ is constant, but they need not start simultaneously. 
\revision{The nodes report their observations to the FC as updates. 
Importantly, since the nodes use monostatic radar and process the received signals locally, the updates only contain position estimates for each target. 
This lessens the need for synchronization in the network; the FC combines detections rather than pulses, so does not require tight timing. 
Further, it is not necessary for the node CPIs to start at the same time and be of the same duration. }

\begin{figure}
    \centering
    \includegraphics[scale=0.5]{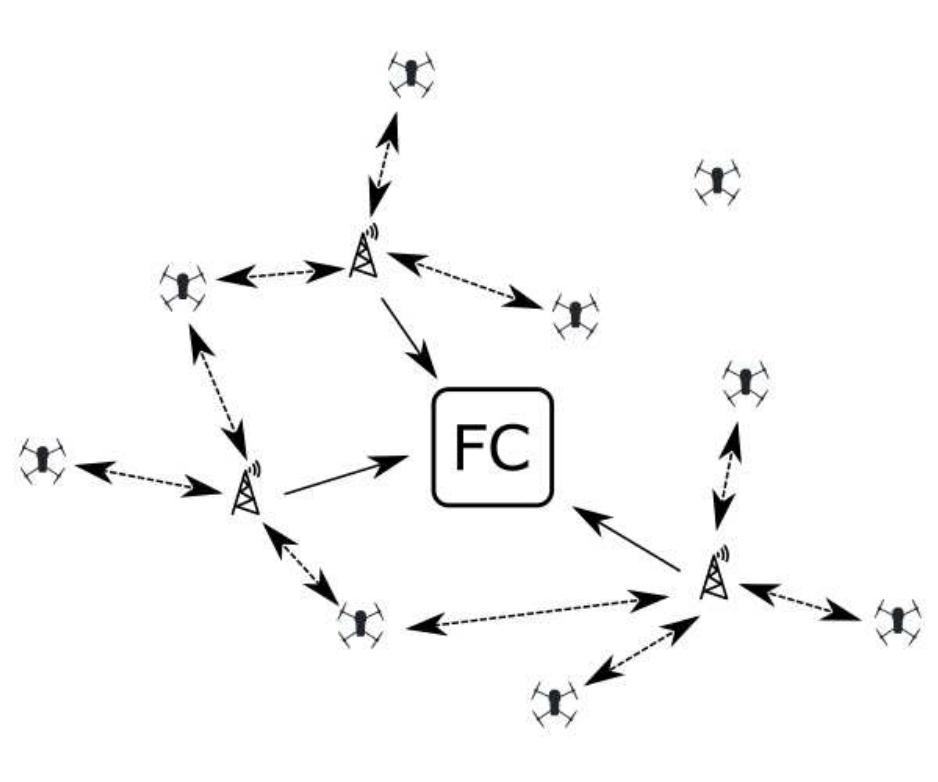}
    \caption{Example network showing which nodes can observe which UAVs. Some targets are observed by multiple nodes and some targets are not observed at all. }
    \label{fig:system}
\end{figure}

\subsubsection{\revision{Measurement Model}}
The FC processes new updates at the same rate that CPIs occur (i.e., 1/CPI). 
All updates are assumed to arrive at the end of the update interval. 
In the distributed approach, the FC processes all updates received during an update interval at the end of the interval.

An update at time $t$ from node $n$ consists of that node's current estimate for each target $m$ in the region $S_n$. 
Based on the target range and aspect, node $n$ observes each target $m$ with a variance $\sigma_{n,m}$ \revision{and a detection probability $P_D$. 
False alarms are generated randomly with fixed probability $P_{FA}$. 
\rev{Note that while the set of existing targets is spatially distributed as a PPP, the set of \emph{detections of those targets} will be spatially non-uniform due to false alarms and missed detections. 
Further, the posterior probability distribution of targets will be spatially non-uniform due to the weights on different motion models. }
Detections are associated with target tracks, and new target tracks are formed for new detections. }
The node $n$ can gather these targets into the set $\hat{\mathcal{X}}_n^{(t)}$. 
Nodes estimate three main quantities for a target $m$: 
\begin{itemize}
    \item The target position $X_m$. 
    \item The target velocity $\dot{X}_m$. 
    \item The target motion model state $\gamma_m$. 
\end{itemize}
Taken together, these quantities form a vector 
\begin{equation}
    \mathbf{X}_m^{(t)} = \left[ X_m^{(t)}, \dot{X}_m^{(t)}, \gamma_m^{(t)}\right]
\end{equation}
which node $n$ estimates, forming 
\begin{equation}
    \hat{\mathbf{X}}_m^{(t)} = \left[ \hat{X}_m^{(t)}, \hat{\dot{X}}_m^{(t)}, \hat{\gamma}_m^{(t)}\right]
    \end{equation}
for all $m\in\hat{\mathcal{M}}^{(t)}_n$ \revision{(the set of targets observed by node $n$ at time $t$)}. 
We use the set $\hat{\mathcal{X}}_n^{(t)}$ to refer to the update provided by node $n$ at time $t$. 
\revision{Denote as $\mathcal{T}^*_{m,n}$ (with $\#(\mathcal{T}^*)\leq T$) the time steps $t<T$ where node $n$ detected target $m$. 
Only targets $m$ such that $\#(\mathcal{T}^*_{m,n})>2$ are included in the update to avoid false tracks. }

The FC forms tracks of all targets it has observed, using Kalman filtering to improve position estimation. 
Denote as 
\begin{equation}
    \overline{\mathbf{X}}^{(t)}_m = \left[ \overline{X}_m^{(t)}, \overline{\dot{X}}_m^{(t)}, \overline{\gamma}_m^{(t)}\right]
\end{equation}
the FC's estimate of target $m$. 
For each $m \in \mathcal{M}^*_n$ \revision{(the set of all targets observed by node $n$ at all times)}, node $n$ forms an Interacting Multiple Model (IMM) filter \cite{blackman1999design}, a type of tracking Kalman filter which evaluates the probabilities of multiple motion models over time. 
This filter estimates the motion model transition probabilities for target $m$ and forms state estimates based on the current estimated motion model. 
Figure \ref{fig:track} shows the fused FC estimate for a single target's track. 
\revision{
Note that if missed detections occur, the node propagates the Kalman filter for the target and provides the predicted values in its update. 
False tracks are avoided by requiring several measurements before a track is confirmed. 
Only target detections are used to update motion model probabilities, not Kalman filter estimates. 
}

\begin{figure}
    \centering
    \includegraphics[scale=0.55]{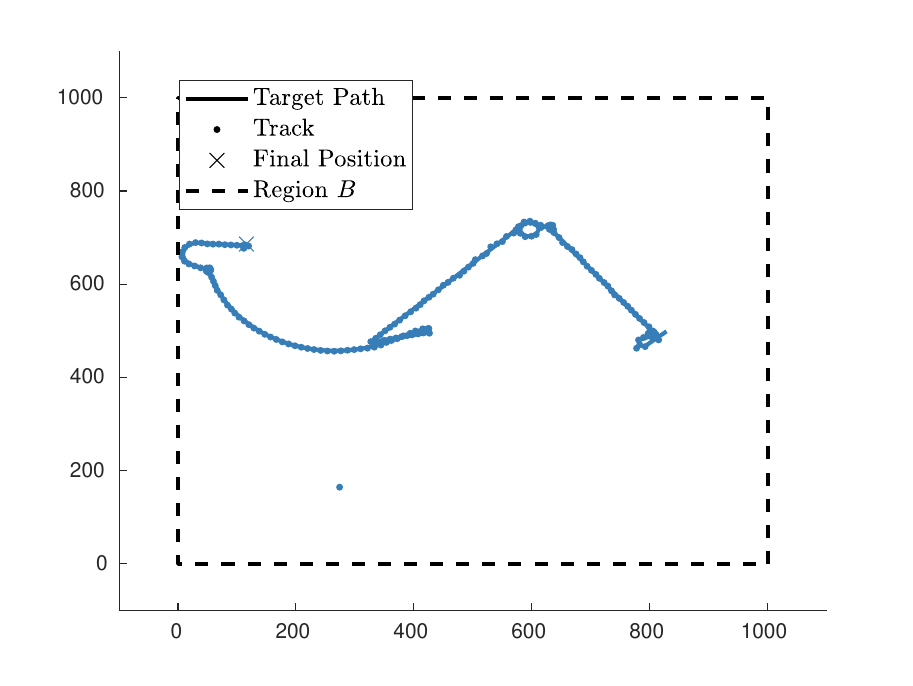}
    \caption{A sample fused track from the perspective of the FC. }
    \label{fig:track}
\end{figure}

Let $\mathcal{N}^{(t)}$ denote the nodes providing updates in the update interval ending at time $t$. 
Then, the FC receives updates on all targets $m\in\mathcal{M}^{(t)}$, where
\begin{equation}
\label{eq:FC_targs}
    \mathcal{X}^{(t)} = \bigcup_{n\in\mathcal{N}^{(t)}}\hat{\mathcal{X}}_n^{(t)}
\end{equation}
Finally, we have set $\mathcal{X}_{FC}^{(t)}$ which is the set of all target tracks formed at the FC. 
Note that $\mathcal{X}_{FC}^{(t)} \supseteq \mathcal{X}^{(t)}$ since each target may not be updated at all $t$. 
\rev{Also note that in a realistic system, the FC would perform track association in order to combine tracks from two nodes which observe the same target. 
Since errors in this process would impact all solutions regardless of the rate at which targets are updated, we omit this step. }
The FC maintains a record of the \emph{age} of each target track. 
Let $v_m$ denote the most recent time for which $m\in\mathcal{X}^{(t)}$. 
Then, the age of the FC track for target $m$ can be written as Eq. (\ref{eq:age}). 
Accordingly, when $m\in\mathcal{X}^{(t)}$, $\Delta_m(t) = 0$, which is the smallest value $\Delta$ can take. 
\begin{equation}
    \label{eq:age}
    \Delta_m(t) = t - v_m
\end{equation}

\subsubsection{\revision{Update Policies}}
The network communicates over a shared communications resource, which is divided into $R$ resource blocks per CPI. 
Other devices are assumed to be present, occupying some amount of the shared resource.
In order to maintain performance and avoid unnecessary interference, the target tracking network is limited to some fraction $f < 1$ of the available spectrum resources. 
Let $C+1 = fR$ be the tracking network's update capacity per CPI. 
Then, $C$ resource blocks are used by the nodes to transmit updates, and one resource block is used by the FC to transmit any feedback. 
Since the problem is trivial if $C>N$ (there is sufficient capacity for each node to provide updates in each CPI), we will assume that $C<N$. 
Prior work \cite{howard2022_MMABjournal} has established techniques for spectrum sharing in CRNs. 
One resource block is sufficient to allow a node to transmit updates on the expected number of targets in each region, $\lambda_m|S_n|$. 
Since the resource blocks must be assigned on-the-fly, we will assume that one resource block is required for a single node to provide any number of target updates. 

Another way of phrasing this, more relevant to Section \ref{sec:distributed}, relates this network rate limit to a limit on the rate of each node. 
If the network update rate is $C$ and there are $\mathbb{E}\left(N\right)=\lambda_n|B|$ nodes, each node is allocated a rate of $\frac{C}{\lambda_n|B|}$ updates per CPI. 
Since $C<N$, let
\begin{equation}
\label{eq:alpha}
    \alpha = \frac{C}{\lambda_n|B|}
\end{equation}
with $0<\alpha<1$ be the (average) number of updates each node can provide per CPI. 
Let $\phi$ be a \emph{policy} which determines the actions that node $n$ takes. 
Then, define an action $\Psi^\phi_n(t)$
\begin{equation}
    \Psi^\phi_n(t) = \begin{cases}
        1, & n\in \mathcal{N}^{(t)}\\
        0, & n\notin \mathcal{N}^{(t)}
    \end{cases}
\end{equation}
which indicates whether node $n$ sends an update during the CPI ending at time $t$. 
\begin{definition}[Constrained Policy]
\label{def:constrained}
A policy $\phi$ is said to be constrained if Eq. (\ref{eq:constrained}) holds. 
\begin{equation}
\label{eq:constrained}
    \limsup_{T\to\infty}\frac1T\mathbb{E}^\phi\left(\sum_{t=0}^{T-1}\Psi^\phi_n(t)\right)\leq\alpha
\end{equation}
\end{definition}

\begin{definition}[Fixed Policy]
\label{def:fixed}
    A policy $\phi$ is said to be \emph{fixed} if $\#(\mathcal{N}_\phi(t))=C\; \forall t$. 
\end{definition}

\begin{lemma}[Fixed Policies are Constrained]
\label{lem:fixed}
Let policy $\phi$ be fixed, so that $\#(\mathcal{N}_\phi(t))=C\; \forall t$. 
Then, $\phi$ is constrained.
\end{lemma}
\begin{proof}
    \begin{align*}
    \limsup_{T\to\infty}\frac1T\mathbb{E}^{\phi}\left(\sum_{t=0}^{T-1}\Psi^{\revision{\phi}}_n(t)\right) &= \limsup_{T\to\infty}\frac1T\mathbb{E}^{\phi}\left(\frac{TC}{N}\right)\\
    &= \limsup_{T\to\infty}\frac1T\left(\frac{TC}{\mathbb{E}(N)}\right)\\
    &= \frac{C}{\lambda_n|B|}\\
    &= \alpha
\end{align*}
so $\phi$ is constrained. 
\end{proof}

The next two sections focus on forming the set $\mathcal{N}^{(t)}$, first in a centralized manner and then allowing each node to decide when to send an update. 


\section{Centralized Policy}
\label{sec:centralized}
Using a polling process, the FC can determine the set $\mathcal{N}^{(t)}$ using the optimization process detailed in this section. 
Specifically, this polling process uses the feedback channel to first poll each node, and then to command that node to provide an update. 
The polling process allows each node to inform the FC if it has observed an ``interesting'' update, which we define as a target entering the environment, exiting the environment, or changing motion model states. 
When a node has an interesting update it is added to the set $\mathcal{A}^{(t)}$.

While the FC wishes to maintain low-age, high-accuracy tracks for all targets, it can only receive updates on these targets by utilizing a set of nodes $\mathcal{N}^{(t)}$ which may provide duplicate updates. 
Generally, due to the spatial coverage model discussed in Section \ref{ss:spatial} and as shown in Figure \ref{fig:nodes_to_targets}, the number of updated targets asymptotically approaches the total number of targets as the number of updated nodes increases. 
Each target may be observed by multiple nodes so as more nodes provide updates, some targets may be updated more than once. 
\begin{figure}
    \centering
    \includegraphics[scale=0.55]{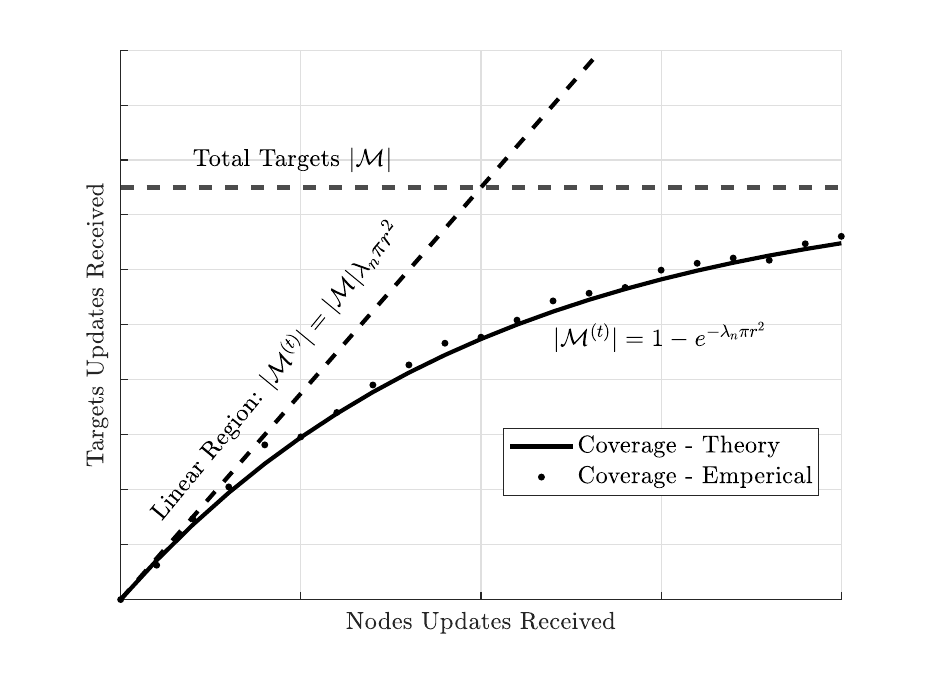}
    \caption{As the update rate for each node increases, the number of targets updated increases more slowly. }
    \label{fig:nodes_to_targets}
\end{figure}
Each target may be updated more than once; as discussed in the coverage model, the number of times any point in the region $B$ is covered is Poisson distributed. 
So, \emph{more than one node may cover any given target} and \emph{each node may cover more than one target}. 
The natural conclusion is that some nodes will provide updates more frequently than others. 
We impose a capacity limit on the centralized network, where $C$ nodes can provide updates per second. 
The centralized approach is a \emph{fixed policy}: the number of nodes selected for updates is constant.

From the perspective of the FC, we have the set of radar nodes $\mathcal{N}$, where each node $n$ can observe each target $m \in \mathcal{M}_{FC}^{(t)}$ with a variance $\sigma_{n,m}$. 
When a node is unable to observe a target, the variance is set to a sufficiently high value. 
This is the target \emph{observability}. 
From the polling process, the FC also has $\mathcal{A}^{(t)}$, the set of nodes with interesting updates. 
This represents the target's \emph{maneuverability}. 

The FC also records the age of each target track for each time $t$ as $\Delta_m(t)$. 
As the FC is only aware of targets $m\in\mathcal{M}_{FC}^{(t)}$, these are the only targets for which age information is maintained. 
Since a target may leave the scene and no longer be observed, we set a maximum age value $\Delta^{max}$ and remove any target with greater age from consideration. 
Formally, 
\begin{equation}
    \mathcal{M}_{FC}^{(t)} = \mathcal{M}_{FC}^{(t-\tau)} \backslash m \;\; \forall m \text{ s.t. } \Delta_m(t) \geq \Delta^{max}
\end{equation}
\revision{Since the FC is not aware of missed detections or false tracks, they are not considered here. }

A naive age- or quality-greedy approach may select the nodes which observe the most maximum-age or minimum-variance targets respectively. 
The centralized objective function we use can be formed as Eq. (\ref{eq:centralized_objective}). 
The product of the target age and node tracking variance forms the base of the objective, with an additional factor of $\alpha\in[0,1]$ when the node in question isn't ``available.'' 
There is an additional discount of $\gamma\in[0,1]$ in the case where target $m$ is not observable by node $n$ $(m\notin S_n$). 

\begin{align}
\label{eq:centralized_objective}
    \mathcal{N}^{(t)}_{cent} &= \max_{\mathcal{N}^{(t)}\in\mathcal{N}} \sum_{n\in\mathcal{N}^{(t)}} Q(n) \\
    \text{s.t. }&\#(\mathcal{N}^{(t)}_{cent}) = C \nonumber \\
    &Q(n) = \sum_{\mathclap{m\in\mathcal{M}^{(t)}}}\tilde{\alpha}_n F_{n,m} \nonumber \\
    &F_{n,m} = \begin{cases}
        \Delta_m(t)\sigma_{n,m}^{-1},  & m \in \hat{\mathcal{X}}_n^{(t)}\\
        \gamma, & m \notin \hat{\mathcal{X}}_n^{(t)} 
    \end{cases} \nonumber \\
    &\tilde{\alpha_n} = \begin{cases}
        1, & n \in \mathcal{A}(t)\\
        \alpha, & \text{else}
    \end{cases} \nonumber
\end{align}

This objective function optimizes the number of target tracks updated, the FC age of those tracks, and the variance with which each node measures that target. 
In other words, if a target is observed at high variance by one node and low variance by another node, the FC will be more likely to select the low variance observation for updating. 
Also, if one node observes more targets than another, it will be more likely to be selected by the FC. 
The FC can only pick $C$ nodes per update period, so this policy cannot exceed the resource constraint. 
The term $\tilde{\alpha}_n$ indicates whether or not node $n$ has an ``interesting'' (targets entering, exiting, or changing motion model state) update. 
If there is no interesting update at a given node there will be a penalty, but it may still be selected. 

A straightforward technique to solve this objective is to form a bipartite matching problem. 
The FC can form a matrix of edges $\mathcal{E}$ between the two vertex sets $\mathcal{N}$ and $\mathcal{M}_{FC}^{(t)}$. 
Each index $\mathcal{E}_{n,m}$ is set to $\tilde{\alpha}_nF_{n,m}$. 
The FC then fills the set $\mathcal{N}_{cent}^{(t)}$ of $C$ nodes by selecting $n=\max_{n\in\mathcal{N}}Q(n)$ until the set is full. 
As each node is selected, the edges corresponding to the selected targets are set to 0 to maximize the number of target updates received. 
Algorithm \ref{algo:centralized} shows all of the relevant steps. 

Since the optimization is subject to $\#(\mathcal{N}_{cent}^{(t)})=C$, the centralized policy is fixed and therefore by Lemma \ref{lem:fixed} is constrained.

\begin{algorithm}
    Poll each node to form $\mathcal{A}(t)$\\
    Form $\mathcal{E}(t)$, where
    \begin{equation}
        \mathcal{E}(t)_{n,m} = \tilde{\alpha}_nF_{n,m}
    \end{equation}\\
    \For{c=1:C}{
        $$\mathcal{N}_{cent}^{(t)} = \mathcal{N}_{cent}^{(t)} \cup \max_{\mathclap{n\notin\mathcal{N}_{cent}^{(t)}}} Q(n)$$\\
        $$\mathcal{E}(t)_{:,m\in\mathcal{M}_n^{(t)}} = 0$$\\
    }
    Return $\mathcal{N}_{cent}^{(t)}$
\caption{Track-Sensitive AoI Node Selection}
\label{algo:centralized}
\end{algorithm}

%


\section{Distributed Age of Incorrect Information Policy}
\label{sec:distributed}
In the distributed approach, each node is added to the set $\mathcal{N}^{(t)}$ if it provides an update at time $t$, determined by a metric which is internal to each node. 
For several reasons, distributed updates involve greater complexity than centralized selection. 
Since communication is limited in the network, each node is not privy to the observations of other nodes. 
In addition, each node is not aware of the times at which other nodes provide updates. 
This implies that since a target \emph{may} be observed by more than one node, it is not possible to know with certainty the state of the FC. 
In other words, the FC may have more up-to-date information than a single node believes, leading that node to overuse the communication resource and provide inefficient redundant updates. 

A distributed approach is desirable for several reasons. 
Primarily, the FC does not necessarily have all relevant information. 
Since nodes are able to directly observe the target states evolving in time, moving the decision process towards the edge of the network allows more accurate decisions to be made.
As a consequence, more useful updates can be provided to the FC meaning that the spectrum resources can be more efficiently used. 
In this section, we will discuss how a policy can maintain a limit on the average update rate while offloading the updating decision to each node. 

\subsection{Preliminaries}
Rather than rely on the AoI as described above, we adopt the Age of Incorrect Information. 
AoII provides an analytically tractable method for a node to provide updates to a FC without direct control from the FC. 
We make several modifications to extend single-observer AoII to the multi-node scenario.

In general, the AoII for a Markov process $X(t)$ is written as 
\begin{equation}
\label{eq:AoII}
    \Delta_{AoII}(t) = f(t) \times g\left(\mathbf{\hat{X}}(t), \overline{\mathbf{X}}(t)\right)
\end{equation}
where $f(t)$ denotes a \emph{time penalty} function, and $g(\mathbf{\hat{X}}(t), \overline{\mathbf{X}}(t))$ an \emph{information penalty} function. 
The penalty function takes the current state $\mathbf{\hat{X}}(t)$ and the \emph{last updated state} $\overline{\mathbf{X}}(t)$ as inputs, and represents the distance between a node's current state and the last known state of the FC. 
When these two states are equal, the FC's tracking error will be limited by its Kalman filtering, even when it receives new data. 
However, when the target motion model changes, the FC's tracking error will increase. 
For the \emph{distributed} timely target tracking problem, we must develop the idea of a \emph{target state} and \emph{target age} somewhat further. 
We discuss the time and information penalty functions more later.

Due to Assumption \ref{ass:motion}, we know that the true motion of a target $m$ evolves in time according to a two-state Markov model with fixed transition probabilities $P_{m,1}$ and $P_{m,2}$. 
A general model could encompass many more states, but a two-state model is sufficient to describe the dynamics. 
When the transition probabilities are both small, the target's motion will update less frequently. 
We'd like for the motion model to update much less frequently than the duration of a single update period, so that the track quality remains high.

The update frequency is better expressed by the \emph{entropy rate} of the Markov chain. 
The entropy rate describes the rate at which a Markov chain changes states. 
More formally the entropy rate is given as Eq. (\ref{eq:entropy_rate}) for a Markov chain with transition matrix $T$ and stationary distribution $\mu$. 
The stationary distribution describes the asymptotic probability that a Markov chain takes a certain state - the portion of time that the Markov chain spends in that state.  
\begin{equation}
\label{eq:entropy_rate}
    H(T) = -\sum_{i,j}\mu_i T_{i,j}\log_2T_{i,j}
\end{equation}
This is shown in Figure \ref{fig:acc_vs_man}. 
When the entropy rate is low, and once the IMM Kalman filter receives enough samples, it can estimate the target state with high accuracy. 
However, when the entropy rate increases, the best-case error also increases. 
\begin{figure}
    \centering
    \includegraphics[scale=0.55]{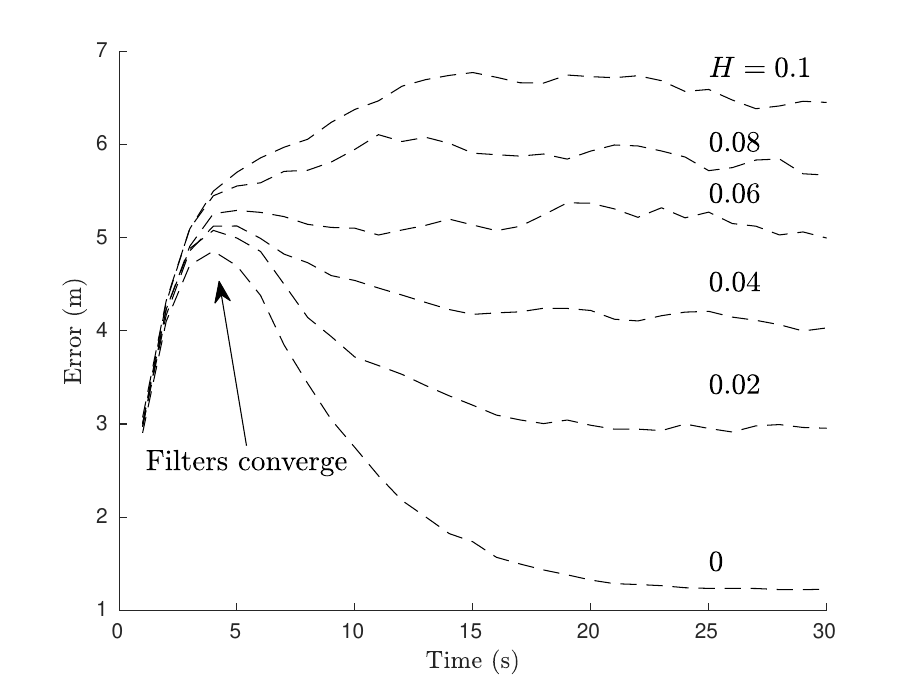}
    \caption{As the entropy rate of a target increases, tracking error decreases if the update rate is held constant. }
    \label{fig:acc_vs_man}
\end{figure}

\subsection{Distributed Rate Limits}
Another naive (but fixed) policy would cause every node to update on a round-robin style schedule (i.e., if the FC update capacity per second is $C$, each node $n\in\mathcal{N}$ sends an update every $\frac{C}{N}$ seconds). 
This would evenly distribute the capacity through the network. 
However, as discussed in the previous section, some targets (and therefore some nodes) will require more frequent updates than others. 
The question then becomes: How should the capacity be distributed through the network? 

The Age of Incorrect Information \cite{AgeOfIncorrectInformation} gives us some direction. 
AoII provides a policy $\phi$ which determines when an observer should send a target update to an aggregator, assuming one of each object. 

The AoII formulation assumes that information from a single process must be collected by a single observer for transmission to an aggregator. 
However, in our problem, information from multiple processes are collected by multiple observers and efficiently relayed to the FC. 
So, there are two fundamental differences between our problem and AoII that must be accounted for. 
Firstly, each node observes possibly multiple targets. 
This can be addressed by extending the one-target Markov motion model to a multi-target model. 
Secondly, more than one node might make observations of a single target. 
Therefore, each node will be unaware of the true state of the FC and so might provide more updates than necessary. 
We address this by using the FC channel to communicate to each node \revision{the targets for which} it is ``responsible''. 

\subsubsection{\revision{Target Assignment}}
After time $t$, the FC has received updates on targets $\mathcal{M}^{(t)}$, Eq. (\ref{eq:FC_targs}). 
The FC can determine $\mathcal{N}_m = \{n \; \text{s.t.} \; m\in\hat{\mathcal{X}}_n^{(t)}\}$ for each target $m\in\mathcal{M}^{(t)}$. 
This is the set of nodes which can track target $m$. 
Then, 
\begin{equation}
\label{eq:closest_node}
    n^*_m = \argmin_{n\in\mathcal{N}_m}\left(||X_m - X_n||_2\right)
\end{equation}
represents the closest of the nodes which can see target $m$. 
Finally, the FC can use the feedback channel to inform node $n$ of all targets $m$ with $n^*_m=n$. 
Define $\mathcal{M}^*_n = \{m \; \text{s.t.} \; n^*_m=n\}$ to be the set of all targets $m$ for which node $n$ is the closest node able to observe $m$. 

Let a node $n$ observe targets $\hat{\mathcal{M}}_n$ with $\#(\hat{\mathcal{M}}_n) = M_n$. 
Let $s$ be a Markov motion model state (i.e., constant-velocity, constant-turn, etc. ). 
Clearly, $\mathcal{M}^*_n \subseteq \hat{\mathcal{M}}_n$. 
Also, note that a Markov chain with states $\{s^1_1, s^1_2, \dots, s^1_w\}$ and a Markov chain with states $\{s^2_1, s^2_2, \dots, s^2_z\}$ can be combined to form a Markov chain with $w\cdot z$ states: 
\begin{equation*}
    \{(s^1_1,s^2_1), (s^1_1,s^2_2), \dots, (s^1_1,s^2_z), (s^1_2,s^2_1), \dots, (s^1_w,s^2_z)\}
\end{equation*} 
Also note that the probability of transitioning from state $(s^1_1,s^2_1)$ to state $(s^1_2,s^2_5)$ is $P_{s^1_1, s^1_2}P_{s^2_1, s^2_5}$. 

\subsubsection{\revision{Distributed Policy}}
Now, node $n$ can determine 1) the targets for which node $n$ is responsible and 2) the Markov chain which models those targets. 
Recall \revision{the} estimated motion state for target $m$, $\hat{\gamma}_m(t)$. 
Node $n$ can estimate the probability that target $m$ transitions from state $i$ to state $j$ as
\begin{align}
    P_{i,j}^m &= \frac{\sum_{\revision{t\in\mathcal{T}^*_{m,n}}} \delta_{i,j}^m(t)}{\sum_{\revision{t\in\mathcal{T}^*_{m,n}}} \mathbbm{1}_{\hat{\gamma}_m(t)=j}}\\
    \delta_{i,j}^m(t) &= \begin{cases}
        1, & \hat{\gamma}_m(\revision{\max{(\mathcal{T}^*_{m,n}<t)}}) = i \;\text{\&} \;\hat{\gamma}_m(t) = j\\
        0, & \text{else}
    \end{cases}
\end{align}
\revision{where $\mathcal{T}^*_{m,n}$ are time steps where node $n$ detects target $m$. }
Then, node $n$ forms a Markov chain $\hat{\Gamma}_n$ of $2^{\#(\mathcal{M}^*_n)}$ states and can estimate each transition probability as Eq. (\ref{eq:big_trans}). 
\begin{equation}
    \label{eq:big_trans}
    P_{i,j} = \prod_{i=1}^{M_n} P_{i,j}^m
\end{equation}
This estimate improves as $t\to\infty$. 

Let the penalty function Eq. (\ref{eq:AoII}) be defined by
\begin{align}
    f(t) &= t-t^*_n\\
    g(\mathbf{\hat{X}}(t), \overline{\mathbf{X}}(t)) &= \begin{cases}
        0, & \hat{\Gamma}_n(t) = \hat{\Gamma}_n(v_n)\\
        1, & \text{else}
    \end{cases}
\end{align}
recalling that $v_n$ is the most recent time for which $n\in \mathcal{N}^{(t)}$. 
\revision{Since the estimated motion model is not updated when detections are missed, the penalty function Eq. (\ref{eq:AoII}) does not increment and an update is not likely. }
The FC's estimate of the motion model of each target observed by node $n$ can be estimated by $\hat{\Gamma}_n(v_n)$, the estimated motion model at the last time that node $n$ provided an update.

The procedure outlined in \cite{AgeOfIncorrectInformation} provides a threshold policy which is Bellman-optimal for this situation. 
In particular, a threshold $p_0$ is determined\footnote{The optimal threshold is determined by \cite{AgeOfIncorrectInformation}, Algorithm 1: ``Optimal Threshold Finder. ''} as a function of the Markov chain transition probabilities and the update rate constraint. 
When the penalty function Eq. (\ref{eq:AoII}) is equal to $p_0$, an update occurs with a probability\footnote{$A(n)$ is given as \cite{AgeOfIncorrectInformation}, Eq. (32). } $\rho_a$, and when the penalty is $p_0+1$, an update occurs with a probability $\rho_b$. 

\revision{
\begin{equation}
    \rho_a = \frac{\alpha - A(p_0+1)}{A(p_0)-A(p_0+1)}
\end{equation}
\begin{equation}
    \rho_b = \frac{A(p_0)-\alpha}{A(p_0)-A(p_0+1)}
\end{equation}
}

Note that while only the targets in $\mathcal{M}^*_n$ are used to determine update times, updates consist of all targets observed by node $n$, since resource blocks are allocated such that all targets observed by a node may be updated simultaneously. 
This means that the FC may, in fact, have more recent information on a given target than node $n$ is aware. 
However, since there is a one-to-one mapping from each target to its closest node, each node must assume that it is the only one providing updates on each target and that it provides the most accurate updates.

Now, to determine the appropriate update constraint, first note that an AoII policy will not be a fixed rate policy since the FC does not have control over the size of $\mathcal{N}^{(t)}$. 
The best that can be hoped for is a mean of $C$. 
An AoII policy is however constrained, under Definition \ref{def:constrained}.

\begin{lemma}[Constraint Equivalence]
\label{thm:equiv}
    If the AoII constraint is 
    \begin{equation}
        \delta = \frac{\alpha}{M_n}
    \end{equation}
    then the average number of updates will be $C$. 
\end{lemma}
\begin{proof}
    See Appendix \ref{app:lem2}
\end{proof}


\section{Simulations and Analysis}
\label{sec:simulations}

\subsection{Baseline Approaches}
In addition to the Age of Information and Age of Incorrect Information policies above, we also provide baseline results for comparison purposes using a multi-armed bandit algorithm, random selection, and a round-robin style approach. 
Multi-armed bandits are often implemented in this style of iterative, online problem solving. 
We select the Upper Confidence Bound (\textbf{UCB}) due to its general-purpose use \cite{UCB_fischer}. 
UCB considers the variance of each target track, and selects nodes accordingly in an attempt to balance exploration of nodes and exploitation of low-variance tracks. 
This is an example of a track-greedy selection algorithm: there are \emph{many} targets in the environment, and an algorithm which is over-fitted to track variance may be susceptible to poor performance due to neglecting high-age tracks or rarely sampled nodes. 

At each update time $t$, the FC chooses nodes according to Eq. (\ref{eq:UCB_node_selection}), where $N_t(n)$ is the number of times $n$ has been selected before time $t$. 
\begin{equation}
\label{eq:UCB_node_selection}
    \mathcal{N}_{UCB}^{(t)} = \argmin_{\mathcal{N}^{(t)} \in \mathcal{N}} \left[\sum_{m\in\mathcal{M}} F_{n,m} + \sqrt{\frac{\log t}{N_t(n)}} \right]
\end{equation}

The second alternative approach we demonstrate is random node selection. 
In each time step, the FC samples the appropriate number of nodes at random and without replacement. 
This represents a strategy that seeks to evenly spread the network capacity over the nodes; each node will be selected equally often. 
Unfortunately, due to the stochastic nature of the network, some nodes will may require more frequent updating and some less. 

Lastly, the round-robin style approach selects the $C$ nodes which have the oldest updates. 
All of these approaches are fixed policies and therefore are constrained, meeting the condition of Definition \ref{def:constrained}.

\subsection{Results}
\begin{table}
    \centering
    \caption{Simulation Parameters}
    \begin{tabular}{c|c|c}
        Variable & Description & Value \\[0.5ex] 
        \hline
        $\lambda_n$ & Node Density per \qty{}{\km\squared} & 0.2 \\
        $\lambda_m$ & Target Density per \qty{}{\km\squared}& 0.3 \\
        $|B|$ & Simulated Region \revision{Area} & \qty{100}{\km\squared} \\
        \revision{$|S_n|$} & \revision{Observable Region Area} & \revision{\qty{10}{\km\squared} }\\
        $C$ & Update Capacity & 2 \\
        N/A & Averaged Simulations & 120 \\ [1ex] 
    \end{tabular}    
    \label{tab:params}
\end{table}
First, we should note that since each policy is constrained (Def. \ref{def:constrained}), they all meet the resource constraint. 
Figure \ref{fig:selected_nodes} demonstrates the average utilized capacity over many averaged simulations. 
Since the proposed distributed algorithm is not fixed (Def. \ref{def:fixed}), sometimes more or less than $C$ nodes provide updates simultaneously. 
This is acceptable since the network can use on average $C$ resource blocks for updates. 
Note that since $C=2$ and $\lambda_n=0.2$, 
\begin{equation}
    \alpha = \frac{C}{\lambda_n|B|} = 0.1
\end{equation}
is the update rate per node. 
Simulation parameters are listed in Table \ref{tab:params}. 

\begin{figure}
    \centering
    \includegraphics[scale=0.55]{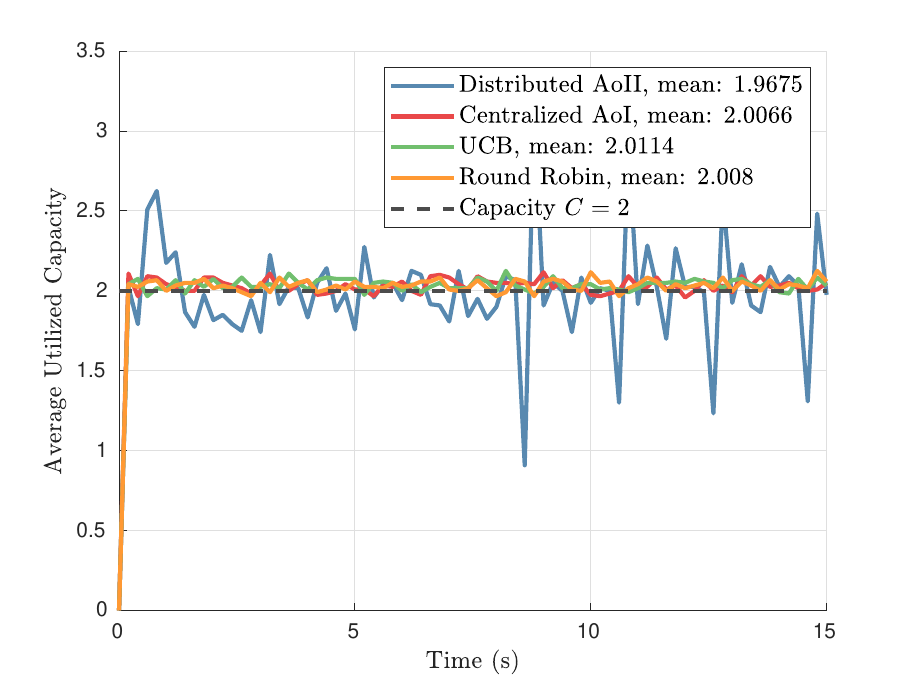}
    \caption{Each algorithm is constrained and therefore meets the average capacity $C=2$. }
    \label{fig:selected_nodes}
\end{figure}

\revision{
\paragraph{Missed Detections and False Alarms}
Radar detection and tracking problems contain an implicit trade-off between missed detections and false alarms. 
Missed detections occur when a target response falls below the detection threshold, and false alarms occur when the clutter or noise response is above the detection threshold. 
Of these, missed detections are more problematic, so detection threshold values which provide $P_{FA}<10^{-4}$ are reasonable \cite{richards2012principles}. 
We show in Figure \ref{fig:pd} that the performance of the distributed AoII policy is only mildly degraded when $P_D=0.9$ and $P_{FA}=10^{-3}$. 
This is because, under the AoII policy, nodes which miss detections at a time step $t$ do not have fresh information and are unlikely to provide an update. 
Since the impact of missed detections and false alarms is minimal, we set $P_D=1$ and $P_{FA}=0$ for the remainder of the results. 
}

\begin{figure}
    \centering
    \includegraphics[scale=0.55]{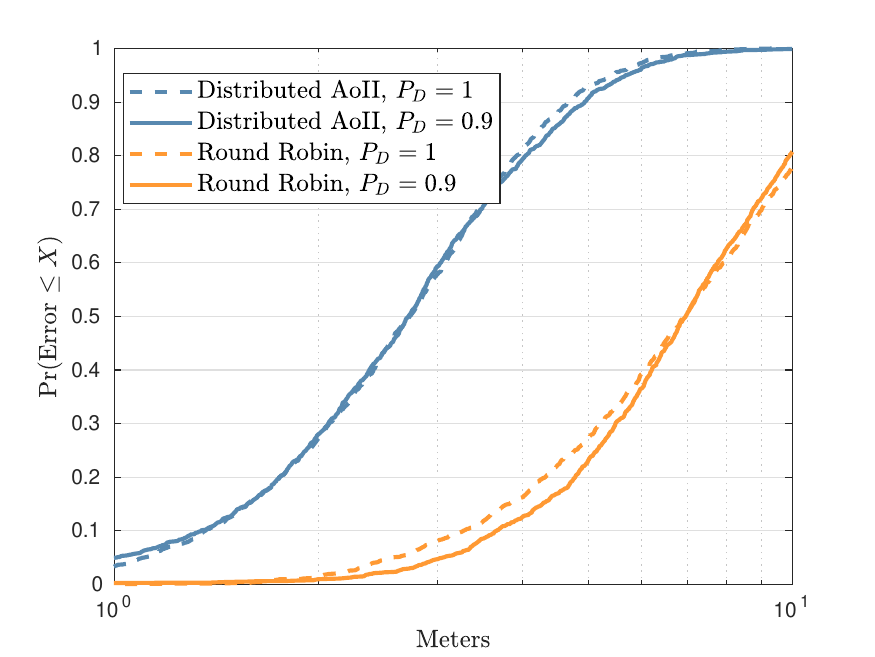}
    \caption{\revision{Tracking performance when the probability of detection is reduced to 0.9 with a capacity of 1. The performance of Round Robin is somewhat degraded, but the performance of the distributed AoII policy does not suffer much. This is because when a node misses a detection, it is much less likely to provide an update (as its own information is less fresh). }}
    \label{fig:pd}
\end{figure}

Each policy is constrained, so the average network utilization rate is constant between policies. 
This does not mean, however, that all targets are updated at the same rate. 
In fact, as Fig. \ref{fig:scatter} shows, the distributed AoII policy updates targets which have higher entropy rates more often, while the round robin policy updates targets at the same rate independent of the motion model entropy rate. 
\revision{
This figure shows several properties. 
First, the AoII policy is able to provide a higher update rate for all targets by allocating more resources to those nodes which observe more targets. 
Second, those targets with higher entropy rate receive more updates than those with lower entropy rates. 
These properties interact to somewhat suppress the effect of higher entropy rate: each node sees multiple targets, entropy rate is uniformly distributed among targets, and nodes seeing more targets get more resources. 
An update rate of ``1'' means that a target is updated by one node per update period. 
}
\begin{figure}
    \centering
    \includegraphics[scale=0.55]{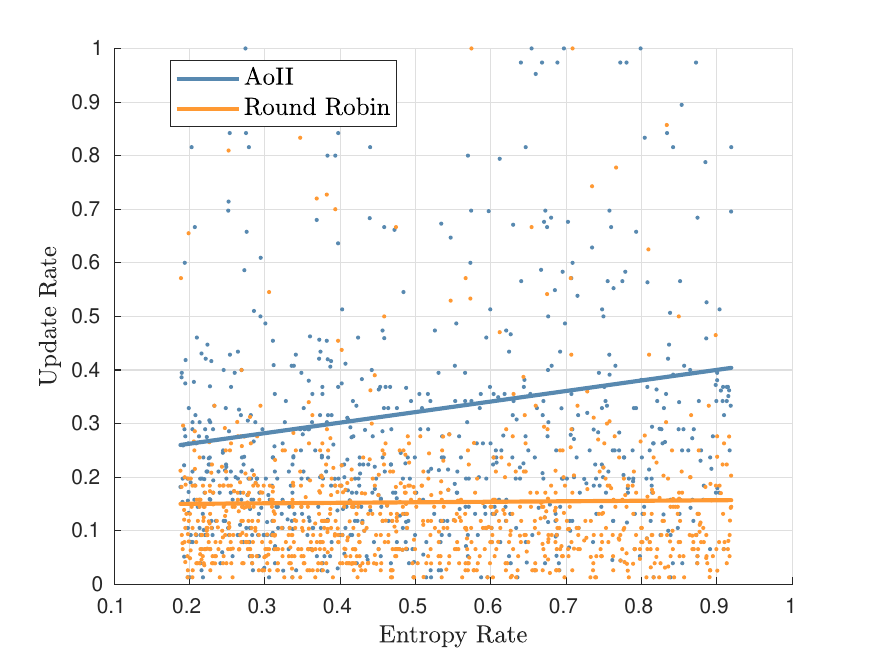}
    \caption{A scatter plot and best-fit line relating the entropy rate of target motion models to the rate at which the FC receives target updates. \revision{Targets are uniformly distributed entropy rates in $[0.2, 0.8]$. }}
    \label{fig:scatter}
\end{figure}

Since each algorithm uses the same number of resource blocks on average, it might be assumed that each algorithm also exhibits the same tracking performance. 
This is not the case, since some algorithms will make more efficient use of the resource blocks, as well as choosing more carefully which nodes provide updates at which times. 
Figure \ref{fig:algo_comp} demonstrates that the UCB and Round Robin algorithms both perform poorly. 
This is due to the reward-greedy behavior exhibited by UCB and the uniform selection probability of Round Robin. 
UCB tends to preferentially select low-variance tracks, and avoid high-variance tracks, due to the reward function formulation. 
This leads to overall \emph{higher} tracking error, since high-variance tracks will age, causing the FC estimate to become even worse.

The AoI-inspired timely algorithm performs better, since it takes into account the age of the tracks reported by each node, and is able to allocate more updates to those nodes which provide more fresh, low-variance information. 
The AoII algorithm performs even better since it removes the timely algorithm's need to explore nodes to discover new tracks; nodes are able to decide when to send their own updates given the update constraint. 

\begin{figure}
    \centering
    \includegraphics[scale=0.55]{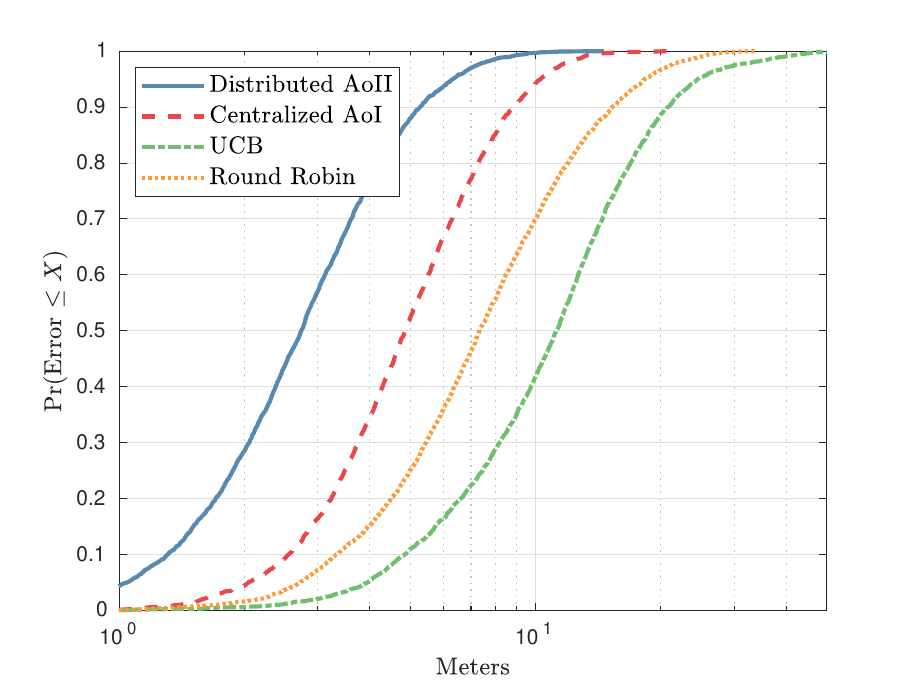}
    \caption{Error distributions for different selection algorithms. Since Round Robin and Random each select nodes with a constant frequency, they exhibit similar performance. 
    Since the centralized AoI and distributed AoII metrics incorporate track age, they outperform the other techniques. 
    Since the AoII approach further allows each node to decide when to provide updates, it achieves nearly double the probability of less than 100 meter accuracy while meeting the same spectrum usage. }
    \label{fig:algo_comp}
\end{figure}

As demonstrated by the above results, efficient use of the available resources results in lower tracking error. 
With this in mind, we next consider how tracking performance varies with changing constraints. 
Fig. \ref{fig:capacity_comp} shows us that increasing the capacity $C$ results in increased tracking performance for both the AoII algorithm and Round Robin. 
However, the performance gap between these two algorithms also varies. 
When the resource constraint is high, the timing of updates matters less since more updates can be transmitted. 
When the resource constraint is low, however, the update timing becomes more important. 
The difference in AoII and Round Robin performance under a low resource constraint is nearly an order of magnitude in tracking error. 
But, when the resource bound is high, they perform nearly identically.

\begin{figure}
    \centering
    \includegraphics[scale=0.55]{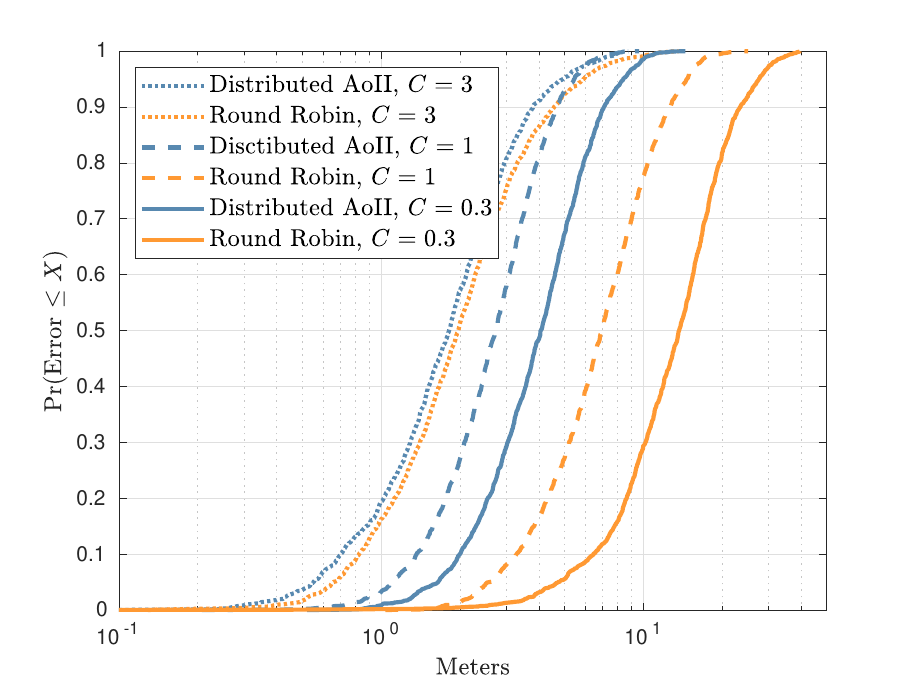}
    \caption{The frequency with which the FC receives updates directly impacts the track error. 
    In addition, the performance gap between AoII and random selection increases with decreasing capacity. 
    When less capacity is available, the selection algorithm quality becomes more important. The capacity constraint $C$ varies from \revision{0.3 to 3}. }
    \label{fig:capacity_comp}
\end{figure}

A critical value for this type of tracking problem is the \emph{peak age of information} (\textbf{PAoI}). 
As opposed to the AoII, this is not a value we used in decision-making, but is a tool for analysis of policies. 
It denotes the average maximum age of a process upon updating - how long the FC must estimate the location of a target before an update is provided. 
We can denote the PAoI as
\begin{equation}
    \Delta^P(t) = \lim_{\tau\to\infty}\frac{1}{U(\tau)}\sum_{n=1}^{U(\tau)}A_u - \mathbb{E}(A_u)
\end{equation}
where there are $U(\tau)$ updates before $t=\tau$, and $A_u$ is the age of the process at the $n^{th}$ update. 

We can expect that a policy which is more egalitarian towards targets will provide a moderate PAoI - each target would get updated evenly often, and so no targets will go extremely long without updating. 
An age-greedy policy, one which selects nodes such that the maximum-age targets are always updated, might provide a lower PAoI. 
However, a policy which takes into account the differing needs of each target may (somewhat paradoxically) provide a \emph{higher} PAoI, if not all targets need similar update rates. 
In Figure \ref{fig:peak_age} we see that the AoII policy outperforms the others in terms of peak age. 
This is because this technique relies on each node to determine when an update is needed rather than rely on the FC. 
The timely policy also performs well, since it optimizes for track age as well as variance. 

\begin{figure}
    \centering
    \includegraphics[scale=0.55]{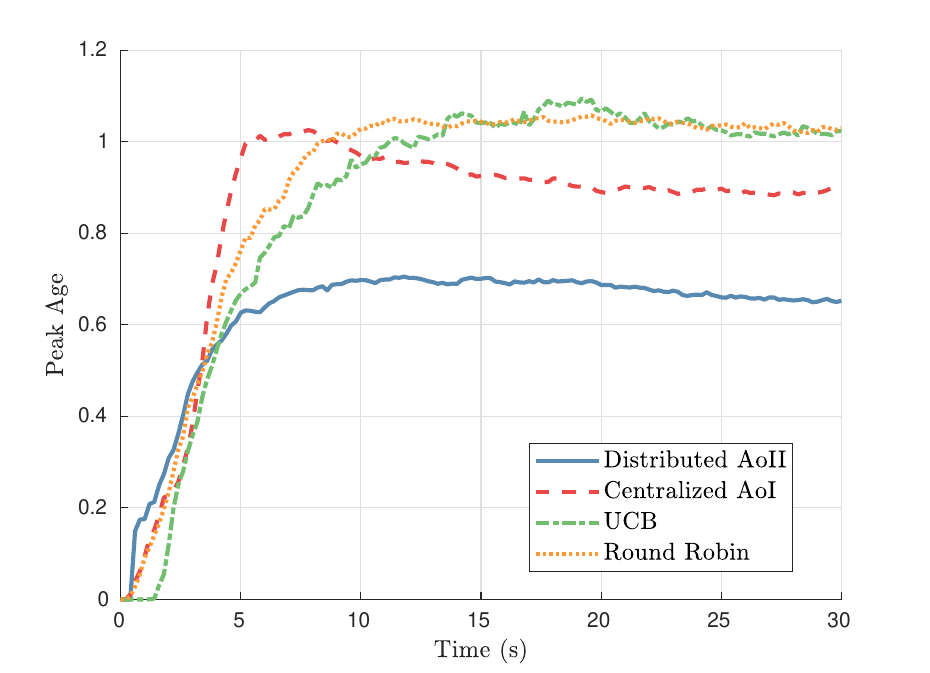}
    \caption{Peak age averaged over all active tracks. The distributed AoII policy exhibits the best performance, followed by the centralized AoI policy. Since the centralized approach first explores nodes which observe many targets, it does not revisit targets before all nodes are updated. }
    \label{fig:peak_age}
\end{figure}

Another way to understand the impact of the capacity constraint is by examining the mean age of active tracks at the FC. 
In Fig. \ref{fig:mean_age}, we see that higher-capacity networks are able to maintain a lower mean age. 
This is simply because of the quantity of targets the FC can update. 

\begin{figure}
    \centering
    \includegraphics[scale=0.55]{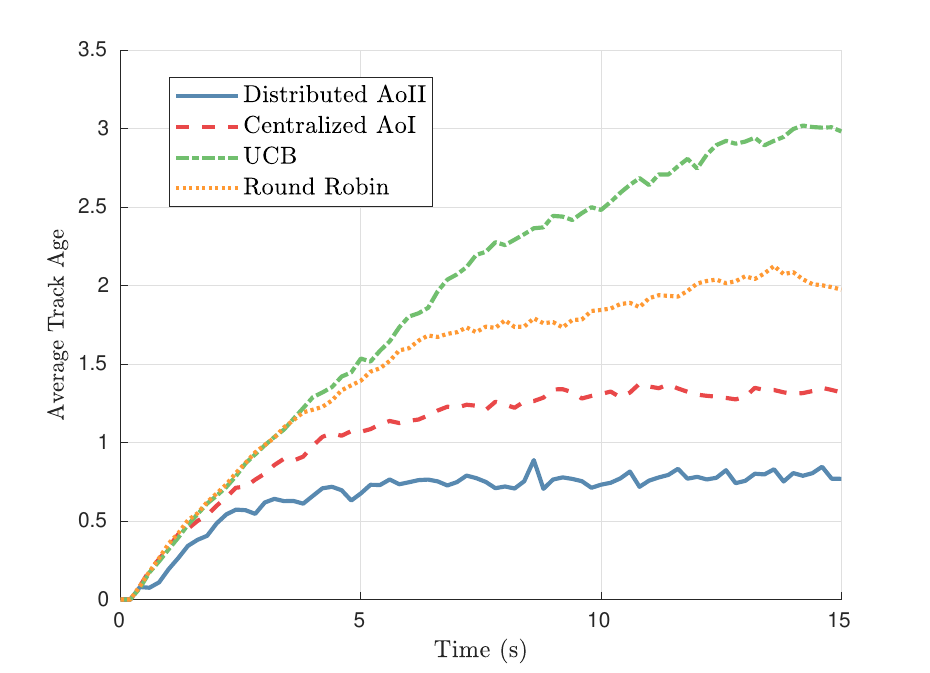}
    \caption{Mean age of each active target track at the FC. 
    Since AoII updates nodes according to track age, it most effectively minimizes average age at the FC. }
    \label{fig:mean_age}
\end{figure}

\revision{W}e can inspect the number of targets which are observable to the network $(m$ s.t. $m \in \cup_{n\in\mathcal{N}}S_n)$ but which are not tracked by the FC, shown in Fig. \ref{fig:total_tracks}. 
This is caused when nodes do not provide updates in a timely manner when new targets appear. 
As the plot demonstrates, the number of missed targets is low for all considered algorithms. 
This is because even those algorithms which underperform in tracking error will still tend to receive updates from all nodes over enough time. 
Notably, the UCB algorithm has the lowest number of untracked targets. 
This is due to UCB rewarding exploration early in the game - it is more likely to select nodes that have not been selected before, so it will initiate a track for all observed targets early in the simulation. 
One key take-away, however, is that UCB's tracking error performance is worse than the AoI based algorithms. 
So, it is not necessarily beneficial to maintain target tracks which receive infrequent updates. 

Lastly from this figure we can see the number of targets which exist in the environment but are out of range of the network (i.e., no radar is able to track them). 
This is a consequence of the random nature of the network. 
Eq. (\ref{eq:unobserved}) shows the probability distribution of the number of targets which are out of range of the network. 
Evaluated on the densities of the network shown, the mean number of unobservable targets is close to 4, which matches the difference between the total targets and the covered targets. 

\begin{figure}
    \centering
    \includegraphics[scale=0.55]{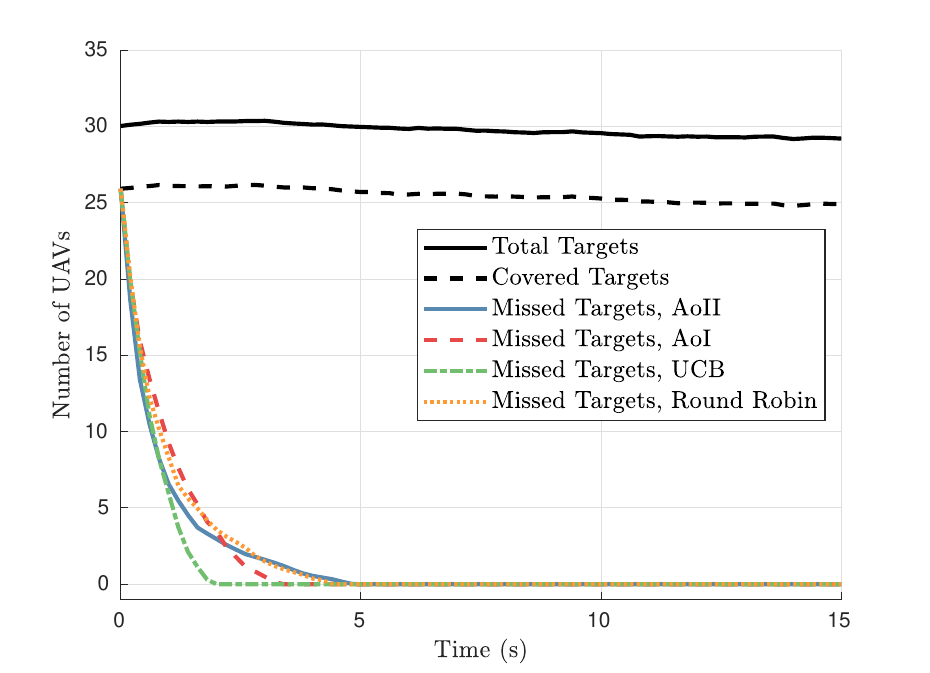}
    \caption{Number of missed targets for each algorithm. Due to the stochastic nature of each network, the number of unobservable targets will be distributed according to Eq. (\ref{eq:unobserved}) which with $\lambda_m=0.3$, $\lambda_n=0.2$ has a mean close to 4. }
    \label{fig:total_tracks}
\end{figure}
%


\section{Conclusions}
\label{sec:conclusions}
In this work, we demonstrated centralized and distributed node selection algorithms for a cognitive radar network which efficiently use the limited communication resources available to the network. 
Each of several radar nodes must observe many targets and efficiently provide updates to the FC. 
The UAV targets exhibit Markov chain motion model behavior, which is exploited to provide more efficient updates. 

We presented a centralized, track-sensitive AoI-inspired node selection algorithm, which utilizes a polling procedure to determine which nodes have observed a change in motion model. 
This method was shown to be superior to alternative techniques such as the UCB multi-armed bandit and random node selection. 

Then, we developed a distributed technique based on the Age of Incorrect Information through which each node can independently determine when to provide updates to the FC. 
This required modification of the AoII criteria developed in \cite{AgeOfIncorrectInformation} to accommodate the multi-target, multi-node scenario. 
We showed a resource constraint that is identical in both distributed and centralized formulations. 
This technique was also shown to be superior to alternative approaches. 

In real target tracking systems, this type of optimization could result in simpler track management and lower communication requirements. 
Due to the lower track age demonstrated by our proposed techniques, tracking performance in realistic systems should be expected to increase, especially for large numbers of targets. 

A possible extension to this work would include ``mode control,'' by which the FC or the nodes could determine how to allocate radar observation time and whether to perform passive sensing, such as direction of arrival estimation or signal classification. 
Such an extension would possibly enable a tracking network to be more power-efficient as well as to reduce the radiated power, allowing the nodes in the network to attain a \revision{lower probability of interception (i.e., lower chance of being detected by other systems).}


\bibliographystyle{IEEEtran}
\bibliography{bibli}

\begin{thebibliography}{10}
\providecommand{\url}[1]{#1}
\csname url@samestyle\endcsname
\providecommand{\newblock}{\relax}
\providecommand{\bibinfo}[2]{#2}
\providecommand{\BIBentrySTDinterwordspacing}{\spaceskip=0pt\relax}
\providecommand{\BIBentryALTinterwordstretchfactor}{4}
\providecommand{\BIBentryALTinterwordspacing}{\spaceskip=\fontdimen2\font plus
\BIBentryALTinterwordstretchfactor\fontdimen3\font minus
  \fontdimen4\font\relax}
\providecommand{\BIBforeignlanguage}[2]{{%
\expandafter\ifx\csname l@#1\endcsname\relax
\typeout{** WARNING: IEEEtran.bst: No hyphenation pattern has been}%
\typeout{** loaded for the language `#1'. Using the pattern for}%
\typeout{** the default language instead.}%
\else
\language=\csname l@#1\endcsname
\fi
#2}}
\providecommand{\BIBdecl}{\relax}
\BIBdecl

\bibitem{howard2023_timelyconf}
W.~W. Howard, C.~E. Thornton, and R.~M. Buehrer, ``Timely target tracking in
  cognitive radar networks,'' in \emph{2023 IEEE Radar Conference
  (RadarConf23)}, 2023, pp. 1--6.

\bibitem{howard2021_multiplayerconf}
W.~W. Howard, C.~E. Thornton, A.~F. Martone, and R.~M. Buehrer, ``Multi-player
  bandits for distributed cognitive radar,'' in \emph{2021 IEEE Radar
  Conference (RadarConf21)}.\hskip 1em plus 0.5em minus 0.4em\relax IEEE, 2021,
  pp. 1--6.

\bibitem{howard2022_MMABjournal}
W.~W. Howard, A.~F. Martone, and R.~M. Buehrer, ``Distributed online learning
  for coexistence in cognitive radar networks,'' \emph{IEEE Transactions on
  Aerospace and Electronic Systems}, pp. 1--14, 2022.

\bibitem{howard2022_adversarialconf}
------, ``Adversarial multi-player bandits for cognitive radar networks,'' in
  \emph{2022 IEEE Radar Conference (RadarConf22)}.\hskip 1em plus 0.5em minus
  0.4em\relax IEEE, 2022, pp. 1--6.

\bibitem{howard2023_hybridjournal}
W.~W. Howard and R.~M. Buehrer, ``Hybrid cognition for target tracking in
  cognitive radar networks,'' \emph{IEEE Transactions on Radar Systems},
  vol.~1, pp. 118--131, 2023.

\bibitem{howard2022_decentralized_conf}
------, ``Decentralized bandits with feedback for cognitive radar networks,''
  in \emph{MILCOM 2022 - 2022 IEEE Military Communications Conference
  (MILCOM)}, 2022, pp. 717--722.

\bibitem{haykin2006}
S.~{Haykin}, ``Cognitive radar networks,'' in \emph{Fourth IEEE Workshop on
  Sensor Array and Multichannel Processing, 2006.}, 2006, pp. 1--24.

\bibitem{Martone_CRN_loop}
A.~F. Martone, K.~D. Sherbondy, J.~A. Kovarskiy, B.~H. Kirk, R.~M. Narayanan,
  C.~E. Thornton, R.~M. Buehrer, J.~W. Owen, B.~Ravenscroft, S.~Blunt,
  A.~Egbert, A.~Goad, and C.~Baylis, ``Closing the loop on cognitive radar for
  spectrum sharing,'' \emph{IEEE Aerospace and Electronic Systems Magazine},
  vol.~36, no.~9, pp. 44--55, 2021.

\bibitem{6218166}
S.~Haykin, ``Cognitive dynamic systems: Radar, {C}ontrol, and {R}adio [point of
  view],'' \emph{Proceedings of the IEEE}, vol. 100, no.~7, pp. 2095--2103,
  2012.

\bibitem{7910111}
A.~F. Martone, ``Cognitive radar demystified,'' \emph{URSI Radio Science
  Bulletin}, vol. 2014, no. 350, pp. 10--22, 2014.

\bibitem{van1968detection}
H.~Van~Trees, K.~Bell, and Z.~Tian, \emph{Detection Estimation and Modulation
  Theory, Part I: Detection, Estimation, and Filtering Theory}, ser. Detection
  Estimation and Modulation Theory.\hskip 1em plus 0.5em minus 0.4em\relax
  Wiley, 1968, no. pt. 1.

\bibitem{YAN2020173}
J.~Yan, W.~Pu, S.~Zhou, H.~Liu, and Z.~Bao, ``Collaborative detection and power
  allocation framework for target tracking in multiple radar system,''
  \emph{Information Fusion}, vol.~55, pp. 173--183, 2020.

\bibitem{9133517}
J.~Yan, W.~Pu, S.~Zhou, H.~Liu, and M.~S. Greco, ``Optimal resource allocation
  for asynchronous multiple targets tracking in heterogeneous radar networks,''
  \emph{IEEE Transactions on Signal Processing}, vol.~68, pp. 4055--4068, 2020.

\bibitem{YAN2022104}
J.~Yan, H.~Jiao, W.~Pu, C.~Shi, J.~Dai, and H.~Liu, ``Radar sensor network
  resource allocation for fused target tracking: A brief review,''
  \emph{Information Fusion}, vol. 86-87, pp. 104--115, 2022.

\bibitem{985685}
F.~Zhao, J.~Shin, and J.~Reich, ``Information-driven dynamic sensor
  collaboration,'' \emph{IEEE Signal Processing Magazine}, vol.~19, no.~2, pp.
  61--72, 2002.

\bibitem{4205087}
J.~Liu, M.~Chu, and J.~E. Reich, ``Multitarget tracking in distributed sensor
  networks,'' \emph{IEEE Signal Processing Magazine}, vol.~24, no.~3, pp.
  36--46, 2007.

\bibitem{AgeOfIncorrectInformation}
A.~Maatouk, S.~Kriouile, M.~Assaad, and A.~Ephremides, ``The age of incorrect
  information: A new performance metric for status updates,'' \emph{IEEE/ACM
  Transactions on Networking}, vol.~28, no.~5, pp. 2215--2228, 2020.

\bibitem{5074349}
S.~Haykin, Y.~Xue, and T.~N. Davidson, ``Optimal waveform design for cognitive
  radar,'' in \emph{2008 42nd Asilomar Conference on Signals, Systems and
  Computers}, 2008, pp. 3--7.

\bibitem{thornton2020efficient}
C.~E. Thornton, R.~M. Buehrer, and A.~F. Martone, ``Efficient online learning
  for cognitive radar-cellular coexistence via contextual thompson sampling,''
  in \emph{GLOBECOM 2020 - 2020 IEEE Global Communications Conference}, 2020,
  pp. 1--6.

\bibitem{kirk2017cognitive}
B.~H. Kirk, J.~W. Owen, R.~M. Narayanan, S.~D. Blunt, A.~F. Martone, and K.~D.
  Sherbondy, ``Cognitive software defined radar: waveform design for clutter
  and interference suppression,'' in \emph{Radar Sensor Technology XXI}, vol.
  10188.\hskip 1em plus 0.5em minus 0.4em\relax SPIE, 2017, pp. 446--461.

\bibitem{thornton2022_universaljournal}
C.~E.~Thornton, R.~M. Buehrer, H.~S.~Dhillon, and A.~F.~Martone, ``Universal
  learning waveform selection strategies for adaptive target tracking,''
  \emph{IEEE Transactions on Aerospace and Electronic Systems}, pp. 1--17,
  2022.

\bibitem{kirk2018avoidance}
B.~H. Kirk, R.~M. Narayanan, K.~A. Gallagher, A.~F. Martone, and K.~D.
  Sherbondy, ``Avoidance of time-varying radio frequency interference with
  software-defined cognitive radar,'' \emph{IEEE Transactions on Aerospace and
  Electronic Systems}, vol.~55, no.~3, pp. 1090--1107, 2018.

\bibitem{ravenscroft2018experimental}
B.~Ravenscroft, J.~W. Owen, J.~Jakabosky, S.~D. Blunt, A.~F. Martone, and K.~D.
  Sherbondy, ``Experimental demonstration and analysis of cognitive spectrum
  sensing and notching for radar,'' \emph{IET Radar, Sonar \& Navigation},
  vol.~12, no.~12, pp. 1466--1475, 2018.

\bibitem{kovarskiy2020spectral}
J.~A. Kovarskiy, J.~W. Owen, R.~M. Narayanan, S.~D. Blunt, A.~F. Martone, and
  K.~D. Sherbondy, ``Spectral prediction and notching of {RF} emitters for
  cognitive radar coexistence,'' in \emph{2020 IEEE International Radar
  Conference (RADAR)}.\hskip 1em plus 0.5em minus 0.4em\relax IEEE, 2020, pp.
  61--66.

\bibitem{martone2017adaptable}
A.~Martone, K.~Gallagher, K.~Sherbondy, A.~Hedden, and C.~Dietlein, ``Adaptable
  waveform design for enhanced detection of moving targets,'' \emph{IET Radar,
  Sonar \& Navigation}, vol.~11, no.~10, pp. 1567--1573, 2017.

\bibitem{thornton2023icassp}
C.~E. Thornton, W.~W. Howard, and R.~M. Buehrer, ``Online learning-based
  waveform selection for improved vehicle recognition in automotive radar,'' in
  \emph{ICASSP 2023 - 2023 IEEE International Conference on Acoustics, Speech
  and Signal Processing (ICASSP)}, 2023, pp. 1--5.

\bibitem{6960029}
J.~Yan, B.~Jiu, H.~Liu, B.~Chen, and Z.~Bao, ``Prior knowledge-based
  simultaneous multibeam power allocation algorithm for cognitive multiple
  targets tracking in clutter,'' \emph{IEEE Transactions on Signal Processing},
  vol.~63, no.~2, pp. 512--527, 2015.

\bibitem{7069270}
J.~Yan, H.~Liu, B.~Jiu, B.~Chen, Z.~Liu, and Z.~Bao, ``Simultaneous multibeam
  resource allocation scheme for multiple target tracking,'' \emph{IEEE
  Transactions on Signal Processing}, vol.~63, no.~12, pp. 3110--3122, 2015.

\bibitem{9114775}
A.~F. {Martone}, K.~D. {Sherbondy}, J.~A. {Kovarskiy}, B.~H. {Kirk}, C.~E.
  {Thornton}, J.~W. {Owen}, B.~{Ravenscroft}, A.~{Egbert}, A.~{Goad},
  A.~{Dockendorf}, R.~M. {Buehrer}, R.~M. {Narayanan}, S.~D. {Blunt}, and
  C.~{Baylis}, ``Metacognition for radar coexistence,'' in \emph{2020 IEEE
  International Radar Conference (RADAR)}, 2020, pp. 55--60.

\bibitem{10007921}
T.~D. Ridder, A.~F. Martone, B.~H. Kirk, and R.~M. Narayanan, ``Multiple
  criteria operational reliability performance metric of a metacognitive
  tracking radar,'' \emph{IEEE Transactions on Aerospace and Electronic
  Systems}, pp. 1--10, 2023.

\bibitem{thornton2022cognitive}
C.~Thornton and R.~Buehrer, ``When is cognitive radar beneficial? {I}nsights
  from dynamic spectrum access,'' in \emph{2023 IEEE Radar Conference
  (RadarConf23)}, 2023, pp. 1--6.

\bibitem{5744132}
B.~Ristic, B.-N. Vo, D.~Clark, and B.-T. Vo, ``A metric for performance
  evaluation of multi-target tracking algorithms,'' \emph{IEEE Transactions on
  Signal Processing}, vol.~59, no.~7, pp. 3452--3457, 2011.

\bibitem{4441756}
R.~Mahler, ``Phd filters of higher order in target number,'' \emph{IEEE
  Transactions on Aerospace and Electronic Systems}, vol.~43, no.~4, pp.
  1523--1543, 2007.

\bibitem{1261119}
------, ``Multitarget bayes filtering via first-order multitarget moments,''
  \emph{IEEE Transactions on Aerospace and Electronic Systems}, vol.~39, no.~4,
  pp. 1152--1178, 2003.

\bibitem{5730505}
R.~P.~S. Mahler, B.-T. Vo, and B.-N. Vo, ``Cphd filtering with unknown clutter
  rate and detection profile,'' \emph{IEEE Transactions on Signal Processing},
  vol.~59, no.~8, pp. 3497--3513, 2011.

\bibitem{5259179}
K.~Panta, D.~E. Clark, and B.-N. Vo, ``Data association and track management
  for the gaussian mixture probability hypothesis density filter,'' \emph{IEEE
  Transactions on Aerospace and Electronic Systems}, vol.~45, no.~3, pp.
  1003--1016, 2009.

\bibitem{1710358}
B.-N. Vo and W.-K. Ma, ``The gaussian mixture probability hypothesis density
  filter,'' \emph{IEEE Transactions on Signal Processing}, vol.~54, no.~11, pp.
  4091--4104, 2006.

\bibitem{AoI_initial}
S.~Kaul, R.~Yates, and M.~Gruteser, ``Real-time status: How often should one
  update?'' in \emph{2012 Proceedings IEEE INFOCOM}, 2012, pp. 2731--2735.

\bibitem{AoI_survey2}
R.~D. Yates, Y.~Sun, D.~R. Brown, S.~K. Kaul, E.~Modiano, and S.~Ulukus, ``Age
  of information: An introduction and survey,'' \emph{IEEE Journal on Selected
  Areas in Communications}, vol.~39, no.~5, pp. 1183--1210, 2021.

\bibitem{AoI_Dhillon}
M.~A. Abd-Elmagid and H.~S. Dhillon, ``Average peak age-of-information
  minimization in {UAV}-assisted {IoT} networks,'' \emph{IEEE Transactions on
  Vehicular Technology}, vol.~68, no.~2, pp. 2003--2008, 2019.

\bibitem{AoI_sensor}
I.~Krikidis, ``Average age of information in wireless powered sensor
  networks,'' \emph{IEEE Wireless Communications Letters}, vol.~8, no.~2, pp.
  628--631, 2019.

\bibitem{AoI_multiple}
R.~D. Yates and S.~Kaul, ``Real-time status updating: Multiple sources,'' in
  \emph{2012 IEEE International Symposium on Information Theory Proceedings},
  2012, pp. 2666--2670.

\bibitem{AoII_ISIT}
S.~Kriouile and M.~Assaad, ``Minimizing the age of incorrect information for
  real-time tracking of markov remote sources,'' in \emph{2021 IEEE
  International Symposium on Information Theory (ISIT)}, 2021, pp. 2978--2983.

\bibitem{AoII_infocom}
C.~Kam, S.~Kompella, and A.~Ephremides, ``Age of incorrect information for
  remote estimation of a binary markov source,'' in \emph{IEEE INFOCOM 2020 -
  IEEE Conference on Computer Communications Workshops (INFOCOM WKSHPS)}, 2020,
  pp. 1--6.

\bibitem{haenggi}
M.~Haenggi, \emph{Stochastic Geometry for Wireless Networks}.\hskip 1em plus
  0.5em minus 0.4em\relax Cambridge University Press, 2012.

\bibitem{8289337}
A.~F. Garc\'ia-Fern\'andez, J.~L. Williams, K.~Granstr\"om, and L.~Svensson,
  ``Poisson multi-bernoulli mixture filter: Direct derivation and
  implementation,'' \emph{IEEE Transactions on Aerospace and Electronic
  Systems}, vol.~54, no.~4, pp. 1883--1901, 2018.

\bibitem{rs12223789}
B.~Li, Z.~Gan, D.~Chen, and D.~Sergey~Aleksandrovich, ``{UAV} maneuvering
  target tracking in uncertain environments based on deep reinforcement
  learning and meta-learning,'' \emph{Remote Sensing}, vol.~12, no.~22, 2020.

\bibitem{blackman1999design}
S.~Blackman and R.~Popoli, \emph{Design and Analysis of Modern Tracking
  Systems}.\hskip 1em plus 0.5em minus 0.4em\relax Artech House, 1999.

\bibitem{UCB_fischer}
P.~Auer, N.~Cesa-Bianchi, and P.~Fischer, ``Finite-time analysis of the
  multiarmed bandit problem,'' \emph{Machine Learning}, vol.~47, pp. 235--256,
  05 2002.

\bibitem{richards2012principles}
M.~Richards, J.~Scheer, and W.~Holm, \emph{Principles of Modern Radar: Basic
  principles}.\hskip 1em plus 0.5em minus 0.4em\relax Tes Dee Publishing Pvt.
  Limited, (Published by arrangement), 2012.

\end{thebibliography}

\appendix
\section*{\revision{Proof of Theorem \ref{thm:equiv}}}
\label{app:lem2}
\setcounter{lemma}{1}
\begin{lemma}[Constraint Equivalence]
    If the AoII constraint is 
    \begin{equation}
        \delta = \frac{\alpha}{M_n}
    \end{equation}
    then the average number of updates will be $C$. 
\end{lemma}
\begin{proof}
    Recall that the target density is $\lambda_m$, the node density is $\lambda_n$, and the capacity $C = \alpha\lambda_n|B|$ is the desired number of updates per CPI. 
    Let R be the variable describing the distance from a point $x\in B$. 
    Let $M_n$ be the random variable describing the number of targets closest to any given node. 
    For any given point in $B$, the probability that no nodes are within a radius $r$ is 
    \begin{equation}
        P(N=0) = e^{-\lambda_n \pi r^2}
    \end{equation}
    and therefore that the probability that at least one node is with a radius $r$ is 
    \begin{equation}
    \label{eq:exp_cdf}
        P(N\geq 1) = 1-e^{-\lambda_n \pi r^2}
    \end{equation}
    Then, the PDF of the distance to the nearest node is
    \begin{align}
        f(r) &= \frac{d}{dr}\left(1-e^{-\lambda_n\pi r^2}\right)\\
        &= 2\lambda_n\pi r e^{\lambda_n\pi r^2}
    \end{align}
    The expected value of this distribution can be found as 
    \begin{align}
        \mathbb{E}[R] &= \int_0^\infty r f(r) dr\\
        &= \int_0^\infty \lambda_n \pi r^2 e^{-\lambda_n \pi r^2}
    \end{align}
    which is a Gaussian integral with solution
    \begin{equation}
        \mathbb{E}[R] =\frac{1}{2\sqrt{\lambda_n}}
    \end{equation}
    This, then, is the expected distance from any point to the nearest node. 

    Now, the expected number of targets for which a given node is closest is given as the number of targets inside the disc of half this radius. 
    \begin{align}
        \mathbb{E}[M_n] &= \lambda_m \pi \left(\frac{1}{4\sqrt{\lambda_n}}\right)^2\\
        &= \frac{\pi}{16} \frac{\lambda_m}{\lambda_n}
    \end{align}

    Finally, the probability node $n$ provides an update in the CPI ending at time $t=\tau$ is $\delta$. 
    We have: 
    \begin{align}
        \mathbb{E}[P(n\in\mathcal{N}^{(\tau)})] &= \mathbb{E}[\prod_{i=1}^{M_n} \delta]\\
        &= \mathbb{E}[M_n\delta]\\
        &= \alpha
    \end{align}

\end{proof}

\begin{IEEEbiography}[{\includegraphics[width=1in,height=1.25in,clip,keepaspectratio]{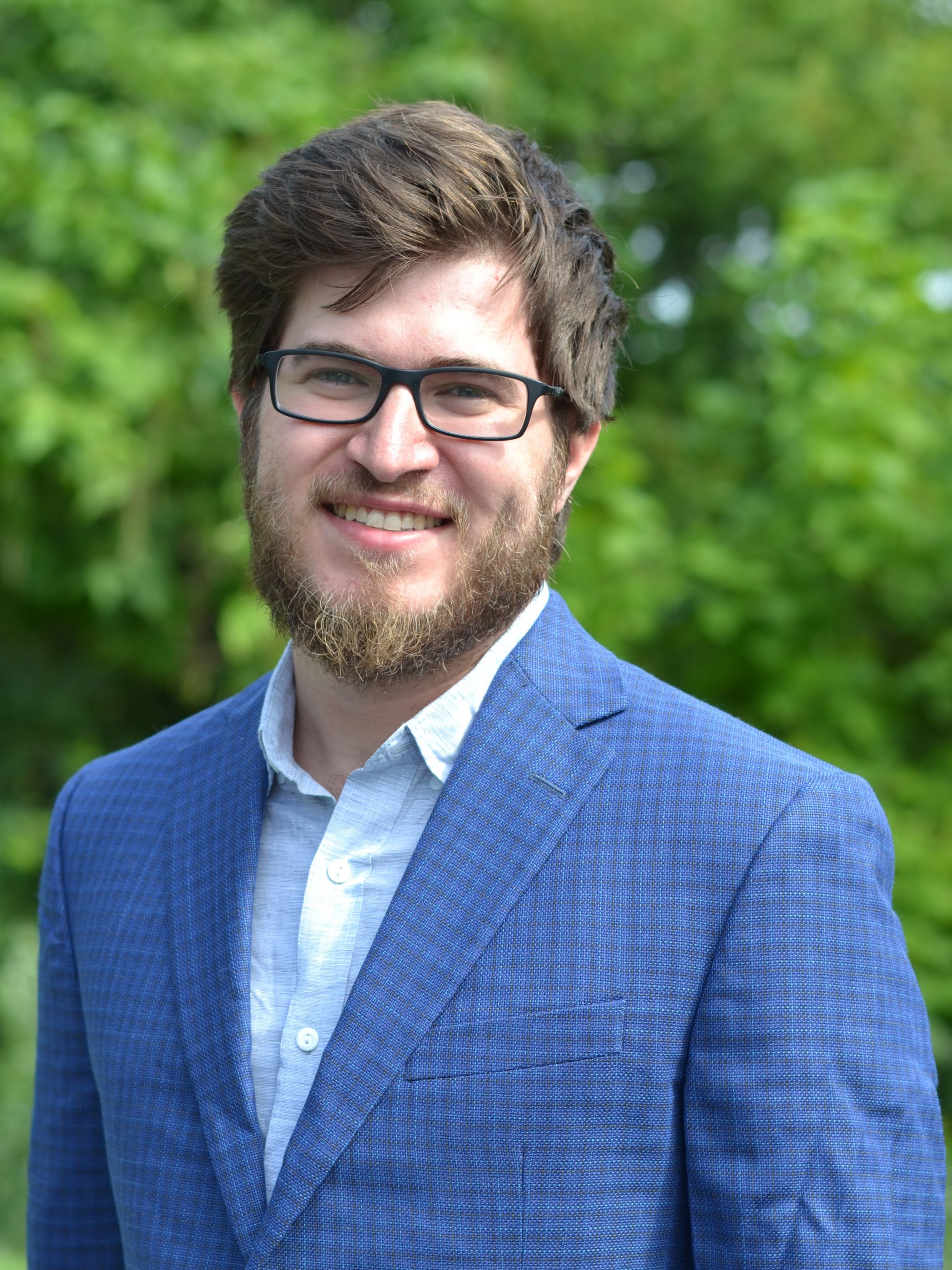}}]%
{William W. Howard}
received the B.S. degree in Electrical Engineering from West Virginia University in 2019 and the M.S and Ph.D. degrees in electrical engineering from Virginia Tech in 2022 and 2023 respectively. 
While at Virginia Tech he was a member of the Wireless@VT research group. 
While in graduate school he worked for Applied Signals Intelligence, Inc. as a signal processing engineer. 
He is currently a research engineer with Rincon Research Corp. in Denver, CO. 
His research interests include array signal processing, direction of arrival estimation, localization, and distributed learning for radar systems. 
In particular, he is interested in the dynamic behavior exhibited by cognitive radar networks, and how the parts of a cognitive radar network can be coordinated. 

He is an active member of IEEE and several societies. 
He has previously served as the IEEE Region 2 Regional Student Representative from 2018 to 2020. 
\end{IEEEbiography}

\begin{IEEEbiography}[{\includegraphics[width=1in,height=1.25in,clip,keepaspectratio]{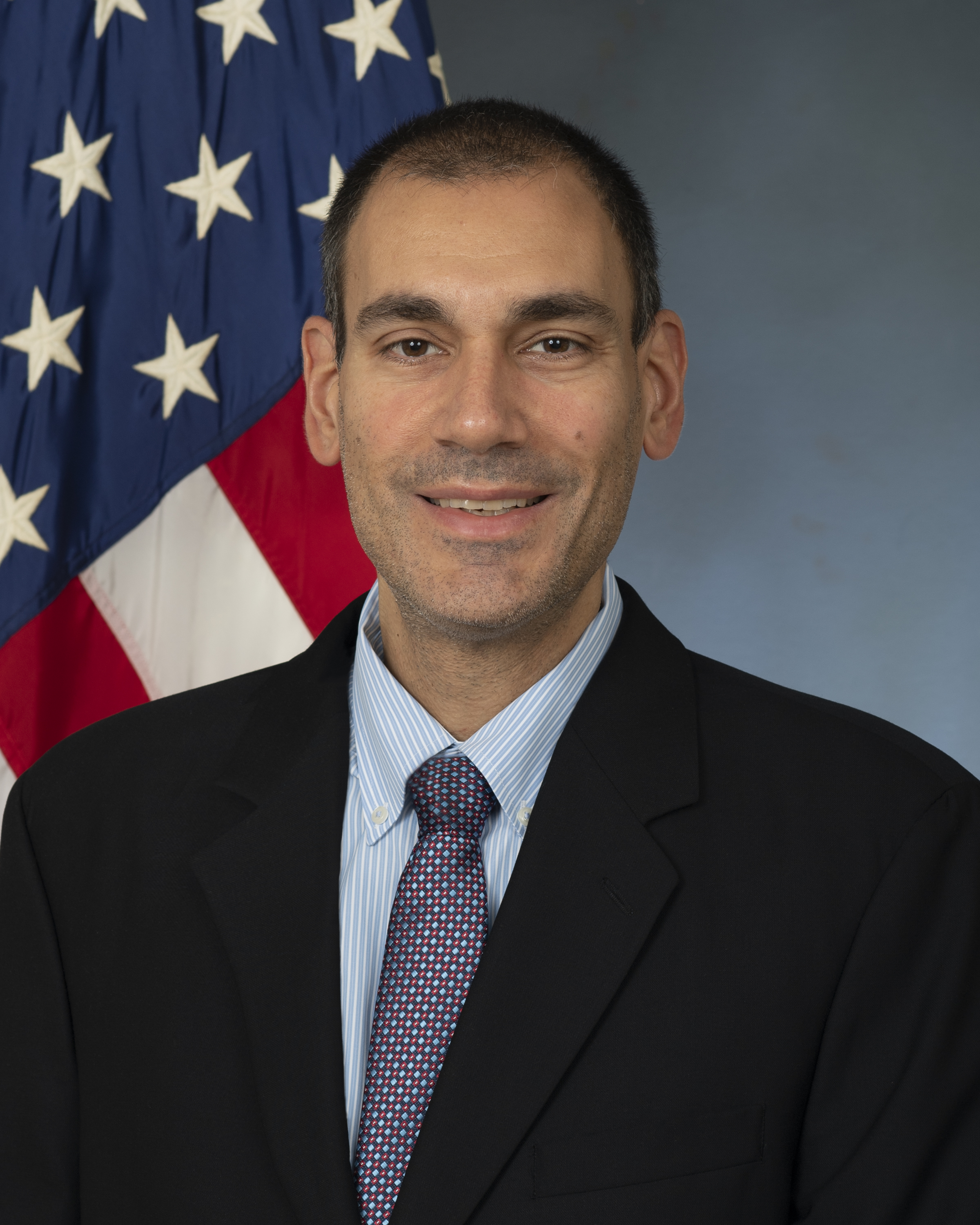}}]%
{Anthony F. Martone}
received the B.S. (summa cum laude) degree in electrical engineering from Rensselaer Polytechnic Institute, Troy, NY, USA, in 2001 and the Ph.D. degree in electrical engineering from Purdue University, West Lafayette, IN, USA, in 2007. He joined the DEVCOM Army Research Laboratory (ARL), Adelphi, MD, USA, in 2007 as a Researcher with the RF Signal Processing and Modeling branch where his research interests include radar, cognitive radar, sensing through the wall technology, spectrum sharing, and radar signal processing. 
He is currently a Subject Matter Expert (SME) for radar spectrum sharing and nonlinear radar technologies with the Department of Defense (DoD) and academic communities. 
He is also leading cognitive radar initiatives with ARL to address spectrum sharing for radar and communication systems, software-defined transceiver control, and adaptive processing techniques. 
He has served as a committee member for eight graduate students with The Pennsylvania State University, the Virginia Polytechnic Institute and State University, and Bowie State University. 
He has authored more than 150 journals and conference publications, two book chapters, ten patents, drafted five new spectrum sharing standards for the IEEE 686 Radar Standards Document, provided the Plenary Presentation at the 2022 IEEE Radar Conference (New York, NY, USA) and served as the General Co-Chair for the 2023 IEEE Radar Conference (San Antonio, TX, USA). 

He was elevated to IEEE Fellow in 2023 for contributions to the development and validation of cognitive radar systems. 
He served as an Associate Editor for \emph{IEEE Transactions on Aerospace and Electronic Systems} from 2017 to 2023, where he received the \emph{IEEE Transactions on Aerospace and Electronic Systems} Associate Editor Award for Excellence in 2023. 
He also served on the Spectral Innovations, Standards, and Publications Committees for the IEEE Aerospace and Electronic Systems Society Radar Systems Panel from 2019 to 2023. 
He received the Commanders Award for Civilian Service in December 2011 for his research and development of sensing through the wall signal processing techniques. 
    
\end{IEEEbiography}

\begin{IEEEbiography}[{\includegraphics[width=1in,height=1.25in,clip,keepaspectratio]{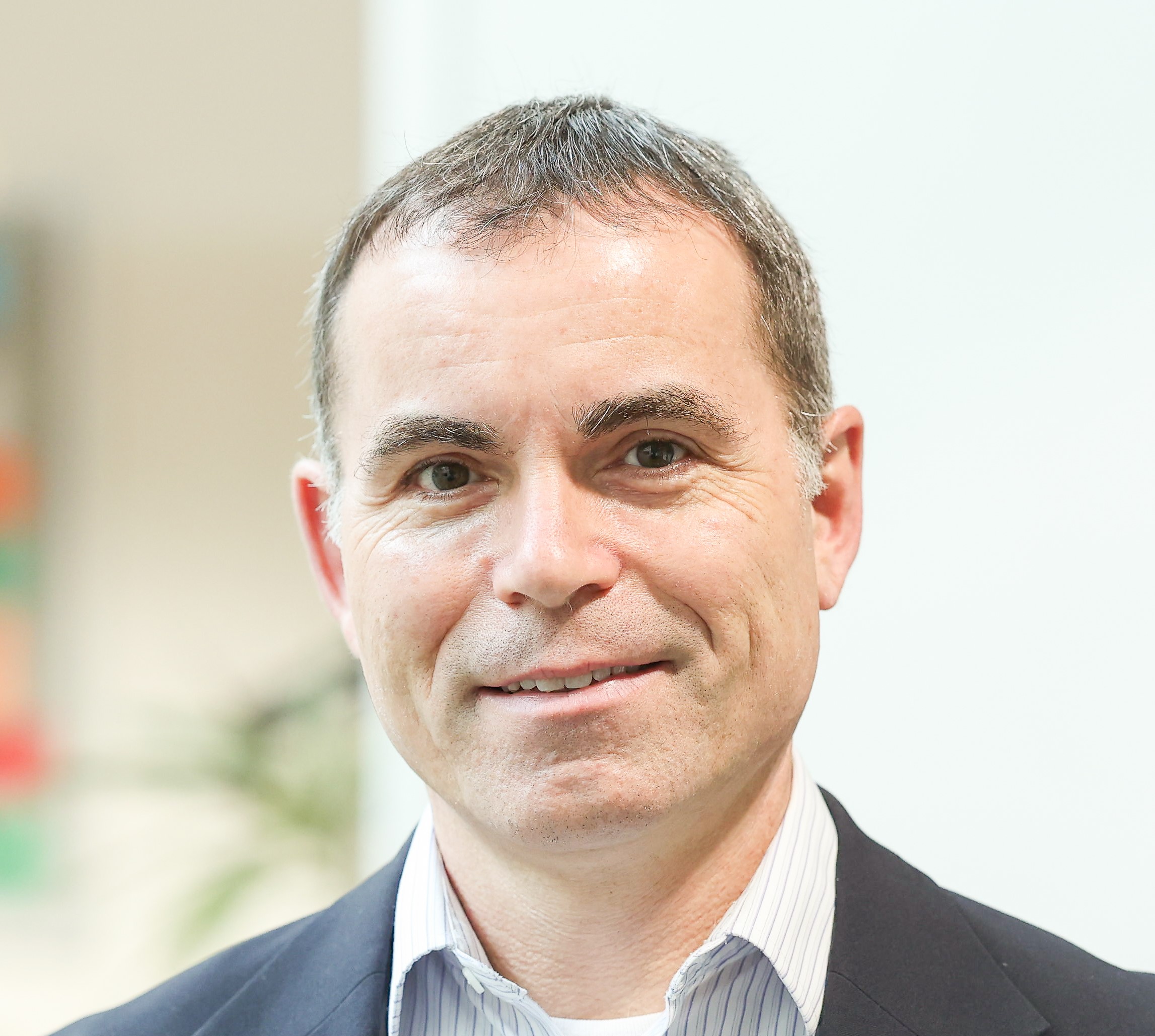}}]%
{R. Michael Buehrer}
joined Virginia Tech from Bell Labs as an Assistant Professor with the Bradley Department of Electrical and Computer Engineering in 2001. 
He is currently a Professor of Electrical Engineering and is the Director of
Wireless @ Virginia Tech, a comprehensive research group focusing on wireless communications, radar and localization. 
During 2009 Dr. Buehrer was a visiting researcher at the Laboratory for Telecommunication Sciences (LTS) a federal research lab which focuses on telecommunication challenges for national defense. 
While at LTS, his research focus was in the area of cognitive radio with a particular emphasis on statistical learning techniques

Dr. Buehrer was named an IEEE Fellow in 2016 “for contributions to wideband signal processing in communications and geolocation.” 
His current research interests include machine learning for wireless communications and radar, geolocation, position location networks, cognitive radio, cognitive radar, electronic warfare, dynamic spectrum sharing, communication theory, Multiple Input Multiple Output (MIMO) communications, spread spectrum, interference avoidance, and propagation modeling. 
His work has been funded by the National Science Foundation, the Defense Advanced Research Projects Agency, the Office of Naval Research, the Army Research Office, the Air Force Research Lab and several industrial sponsors.

Dr. Buehrer has authored or co-authored over 90 journal and approximately 260 conference papers and holds 18 patents in the area of wireless communications. 
In 2021 he was the co-recipient of the \emph{Vanu Bose Award} for the best paper at \emph{MILCOM’21}. 
In 2010 he was co-recipient of the \emph{Fred W. Ellersick MILCOM Award} for the best paper in the unclassified technical program. 
He was formerly an Area Editor for \emph{IEEE Wireless Communications}. 
He was also formerly an associate editor for \emph{IEEE Transactions on
Communications}, \emph{IEEE Transactions on Vehicular Technologies}, \emph{IEEE Transactions on Wireless Communications}, \emph{IEEE Transactions on Signal Processing}, \emph{IEEE Wireless Communications Letters}, \emph{and IEEE Transactions on Education}. 
He has also served as a guest editor for special issues of \emph{The Proceedings of the IEEE}, and \emph{IEEE Transactions on Special Topics in Signal Processing}. 
In 2003 he was named Outstanding New Assistant Professor by the Virginia Tech College of Engineering and in 2014 he received the Dean’s Award for Excellence in Teaching.
\end{IEEEbiography}
\vspace{2.3in}

\end{document}